\documentclass[journal,letterpaper,onecolumn,twoside,nofonttune]{IEEEtran}

\usepackage[utf8]{inputenc} 
\usepackage[T1]{fontenc}
\usepackage{url}
\usepackage{ifthen}
\usepackage{cite}
\usepackage[cmex10]{amsmath}
\usepackage{amsmath}
\usepackage{amsthm}
\usepackage{amsfonts}
\usepackage{amssymb}
\usepackage{mathtools}
\usepackage{mathrsfs}
\usepackage{color}
\usepackage{verbatim}
\usepackage[normalem]{ulem}
\usepackage{enumitem}
\usepackage{hyperref}
\usepackage[ruled,vlined]{algorithm2e}
\usepackage{algpseudocode}
\usepackage{balance}
\usepackage{pifont}
\usepackage{hhline}

\definecolor{purple}{rgb}{0.5, 0.0, 0.5}
\definecolor{dark_green}{rgb}{0.0, 0.5, 0.0}
\definecolor{mygray}{gray}{0.6}
\definecolor{orange}{rgb}{1,0.5,0}

\DeclareMathAlphabet{\mathpzc}{OT1}{pzc}{m}{it}

\newcommand{\blue}{\color{blue}}

\newcommand{\purple}{\color{purple}}
\newcommand{\white}{\color{white}}

\newcommand{\gray}{\color{mygray}}
\newcommand{\Cc}{\mathcal{C}}
\newcommand{\ow}{\mathcal{O}}

\newcommand{\Hh}{\mathcal{H}}
\newcommand{\I}{\mathcal{I}}

\newcommand{\J}{\mathcal{J}}

\newcommand{\Mb}{\bold{M}}

\newcommand{\Wb}{\bold{W}}
\newcommand{\ub}{\bold{u}}

\newcommand{\vb}{\bold{v}}

\newcommand{\Yb}{\bold{Y}}
\newcommand{\yb}{\bold{y}}

\newcommand{\kbar}{\bar{k}}

\newcommand{\Gb}{\bold{G}}
\newcommand{\Gbb}{\bar{\Gb}}
\newcommand{\Gbt}{\tilde{\bold{G}}}

\newcommand{\T}{\mathcal{T}}
\newcommand{\Tbar}{\bar{T}}

\newcommand{\Z}{\mathbb{Z}}
\newcommand{\R}{\mathbb{R}}
\newcommand{\C}{\mathbb{C}}
\newcommand{\F}{\mathbb{F}}

\newcommand{\N}{\mathbb{N}}
\newcommand{\Q}{\mathbb{Q}}
\newcommand{\gb}{\bold{g}}

\newcommand{\Pb}{\bold{P}}
\newcommand{\Sb}{\bold{S}}

\newcommand{\Ab}{\bold{A}}
\newcommand{\Acal}{\mathcal{A}}
\newcommand{\Acalh}{\hat{\Acal}}
\newcommand{\Apzc}{\mathpzc{A}}

\newcommand{\Abt}{\tilde{\bold{A}}}

\newcommand{\bb}{\bold{b}}
\newcommand{\bbh}{\hat{\bold{b}}}

\newcommand{\Bb}{\bold{B}}
\newcommand{\Bbt}{\tilde{\bold{B}}}
\newcommand{\Bcal}{\mathcal{B}}
\newcommand{\Bcalb}{\bar{\mathcal{B}}}
\newcommand{\ab}{\bold{a}}

\newcommand{\cb}{\bold{c}}

\newcommand{\Dbt}{\tilde{\bold{D}}}

\newcommand{\Eb}{\bold{E}}

\newcommand{\eb}{\bold{e}}
\newcommand{\Ib}{\bold{I}}
\newcommand{\Hb}{\bold{H}}

\newcommand{\Kb}{\bold{K}}

\newcommand{\thetah}{\hat{\theta}}

\newcommand{\sfsty}[1]{\ensuremath{\mathsf{#1}}}  
\newcommand{\Enc}{\sfsty{Enc}}

\newcommand{\rsa}{\sfsty{RSA}}
\newcommand{\rs}{\sfsty{RS}}
\newcommand{\brs}{\sfsty{BRS}}
\newcommand{\bch}{\sfsty{BCH}}
\newcommand{\x}{\sfsty{x}}
\newcommand{\err}{\mathrm{err}}
\newcommand{\prp}{\mathrm{PRP}}

\DeclarePairedDelimiter\floor{\lfloor}{\rfloor}

\newtheorem{Thm}{Theorem}
\newtheorem{Cor}[Thm]{Corollary}
\newtheorem{Prop}[Thm]{Proposition}
\newtheorem{Lemma}[Thm]{Lemma}
\newtheorem{Def}[Thm]{Definition}
\newtheorem{Rmk}[Thm]{Remark}

\DeclareMathOperator{\GL}{GL}
\DeclareMathOperator{\nnz}{nnzr}
\DeclareMathOperator{\supp}{supp}
\DeclareMathOperator{\rank}{rank}
\newcommand{\qvec}[1]{\textbf{\textit{#1}}}
\newcommand{\vect}[1]{\boldsymbol{#1}}

\newcommand{\pb}{\vect{p}}
\newcommand{\xmark}{\ding{55}}

\begin{document}

\title{Securely Aggregated Coded Matrix Inversion}

\author{$\textbf{Neophytos Charalambides}^{\mu}$, $\textbf{Mert Pilanci}^{\sigma}$, \textbf{and} $\textbf{Alfred O. Hero III}^{\mu}$\\
$\text{\white.}^{\mu}$EECS Department University of Michigan $\text{\white.}^{\sigma}$EE Department Stanford University\\
  Email: \texttt{neochara@umich.edu}, \texttt{pilanci@stanford.edu}, \texttt{hero@umich.edu}
\vspace{-4mm}
}

\maketitle


\begin{abstract}
Coded computing is a method for mitigating straggling workers in a centralized computing network, by using erasure-coding techniques. Federated learning is a decentralized model for training data distributed across client devices. In this work we propose approximating the inverse of an aggregated data matrix, where the data is generated by clients; similar to the federated learning paradigm, while also being resilient to stragglers. To do so, we propose a coded computing method based on gradient coding. We modify this method so that the coordinator does not access the local data at any point; while the clients access the aggregated matrix in order to complete their tasks. The network we consider is not centrally administrated, and the communications which take place are secure against potential eavesdroppers.
\end{abstract}


\vspace{-1mm}
\section{Introduction and Related Work}
\label{intro}

Inverting a matrix is one of the most important operations in numerous applications, such as, signal processing, machine learning, and scientific computing \cite{GS59,Hig02}. A common way of inverting a matrix is to perform Gaussian elimination, which requires $\ow(N^3)$ operations for square matrices of order $N$. In high-dimensional applications, this can be cumbersome. Over the past few years the machine learning (ML) community has made much progress on \textit{federated learning} (FL), focusing on iterative methods.

The objective of FL is to leverage computation, communication and storage resources to perform distributed computations for ML models, where the data of each federated worker is never shared with the coordinator of the network; that aggregates local computations in order to update the model parameters. In FL applications it is important that the data is kept private and secure.

Distributed computations in the presence of \textit{stragglers} (workers who fail to compute their task or have longer response times) must account for the effect of non-responsive workers. Coding-theoretic approaches have been adopted for this purpose \cite{LLPPR17,YSRKSA18}, and fall under the framework of \textit{coded computing} (CC). Other techniques have also been utilized; to develop \textit{approximate} CC schemes, \textit{e.g.} equiangular tight frames \cite{KSD17} and sketching \cite{CMPH22}. Data security is also an increasingly important issue in CC \cite{LA20}. Despite the fact that multiplication algorithms imply inversion algorithms and vice versa, in the context of CC; matrix inversion has not been studied as extensively as \textit{coded matrix multiplication} (CMM) \cite{YMAA17}. The main reason for this is the fact that the latter is non-linear and non-parallelizable as an operator. We point out that distributed inversion algorithms do exist, though these make assumptions on the matrix, are specific for distributed and parallel computing platforms, and require a matrix factorization; or heavy and multiple communication instances between the workers and the coordinator.

In \cite{CPH22a} a CC method\footnote{We abbreviate `coded computing method/methods' to CCM/CCMs.} was proposed based on \textit{gradient coding} (GC) \cite{HASH17}, which approximates the inverse of a matrix $\Ab$. In order to overcome the obstacle of non-linearity, the columns of $\Ab^{-1}$ are approximated. When assuming infinite floating-point precision, this CCM introduces no numerical nor approximation errors. Note that GC and not CMM was utilized, as the latter does not require the encoding to be done locally by the workers.

Though the two areas of FL and CC seem to be closely related, on the surface they appear incompatible. For instance, in CC one often assumes there is a master server that distributes the data and may perform the encoding (encoding by the master server is done in CMM, but not in GC), while in FL the central coordinator never has access to the distributed local training data; which are located at different client nodes or workers.

There are a few recent works that leverage CC in order to devise secure FL methods for distributed regression and iterative optimization \cite{DPYTH19,PDAYTAH20,SKRA21,SRRA22,XARWZ22,HZSK19}. In this work, we combine optimization and CC, using erasure coding to protect against stragglers as in CC and locally approximating the inverse without revealing the data to the coordinator, to design a CCM which inverts a matrix from data aggregated through clients; and guarantees security against eavesdroppers. Our approach, is based on the \textit{coded matrix inversion method} (CMIM) we develop, which utilizes \textit{balanced Reed-Solomon} ($\brs$) codes \cite{HLH16,HLH16b}. This results in an efficient decoding in terms of the threshold number of responsive workers needed to  perform an error free computation. We show that the general class of \textit{maximum distance separable} (MDS) generator matrices could be used to generate a suitable erasure code (Theorem \ref{MDS_CC_thm}). The focus is on $\brs$ codes, which have the following advantages:
\begin{enumerate}[label=(\roman*)]
  \item minimum redundancy per task across the network,
  \item they optimize communication from workers to the master,
  \item we can efficiently decode the resulting method.
\end{enumerate}

As noted in \cite{DPYTH19,SKRA21}, most CCMs are not applicable in FL. In our case, the obstacle is that all clients need to know each others data in order to invert the aggregated matrix, which we elaborate on in \ref{know_of_A_subs}. For this reason, we relax the privacy restriction of FL and allow the clients to recover the aggregated matrix $\Ab$, which is necessary and unavoidable for matrix inversion.

Our CMIM can also be used to compute the Moore–Penrose pseudoinverse $\Yb^{\dagger}$ of a data matrix $\Yb\in\R^{M\times N}$ for $M\gg N$, which is more general than inverting a square matrix. By using the fact that $\Yb^{\dagger}=(\Yb^\top\Yb)^{-1}\Yb^\top$, the bottleneck is computing the inverse of $\Ab=\Yb^\top\Yb$. In addition, two more matrix multiplications need to take place distributively: computing $\Ab$ before the inversion; and $\widehat{\Ab^{-1}}\Yb^\top$ after the inverse has been approximated. The matrix products can be computed distributively using various CCMs, \textit{e.g.} we can use a modification of the coded FL approaches of \cite{SKRA21} and a CMM from \cite{CMH20}; both of which are based on GC. For the remainder of the paper, we focus on the generic problem of inverting a square matrix $\Ab$.

The proposed approach applies to general linear regression problems. Compared to traditional FL iterative approaches \cite{{KMRR16}}, the difference is that for $\Yb\theta=\pb$; with $\pb$ the label vector and $\theta$ the model parameters, the pseudoinverse-regularized regression solution is $\thetah=\widehat{\Yb^{\dagger}}\pb$.
Unlike conventional FL methods, this regularized regression can be computed non-iteratively. The non-iterative nature of the proposed approach is advantageous in settings such as Kalman filtering, where the matrix inverse must be updated in real time as measurements come in, as well as when dealing with time-series; and regularized regression with varying regularized coefficients.

We organize the paper as follows. In \ref{prel_sec} we recall basic facts on matrix inversion, least squares approximation and finite fields. In \ref{BRS_sec} we review $\brs$ codes, and prove two key lemmas regarding their generator matrices. In \ref{Appr_alg_sec} we present the matrix inverse approximation algorithm we utilize in our CCM. The main contribution is presented in \ref{FCMI_sec}. Our approach is split into four phases, which we group in pairs of two. First, we discuss information sharing from the coordinator to the workers (we consider all the clients' servers as the network's workers), and then information sharing between the workers. Second, we show how our inversion algorithm can be incorporated in linear CCMs, and describe how this fits into the relaxed FL setting we are considering. Concluding remarks and future work are presented in \ref{concl_sec}.

\subsection{Overview of the Coded Matrix Inversion Method}

In CC the computational network is centralized, and is comprised of a master server who communicates with $n$ workers. The idea behind our approximation is that the workers use a least squares solver to approximate multiple columns of $\Ab^{-1}$, resulting in a set of local approximations to submatrices of $\widehat{\Ab^{-1}}$, which we refer to as \textit{blocks}. We present approximation guarantees and simulation results for \textit{steepest descent} (SD) and \textit{conjugate gradient} (CG) iterative optimization methods. By locally approximating the columns in this way, the workers can linearly encode the blocks of $\widehat{\Ab^{-1}}$. The clients have a block of data $\{\Ab_\iota\}_{\iota=1}^k$, which constitute the data matrix $\Ab=\big[\Ab_1 \ \cdots \ \Ab_k\big]$. To simplify our presentation, we assume that each local data block is of the same size; \textit{i.e.} $\Ab_{\iota}\in\R^{N\times T}$ for $T=N/k$, and that client $i$ has $n_i$ servers. Therefore, the total number of workers is $n=\sum_{j=1}^kn_j$. We assume the blocks are of the same size, so that the encodings carried out by the clients are consistent. In \ref{FCMI_sec}, we show that this assumption is not necessary. Moreover, for the CCM, it is not required that the number of blocks equal the number of clients. For a given natural number $\gamma$, assume that $\gamma$ divides $T$; denoted $\gamma\mid T$ (each local data block $\Ab_\iota$ is further divided into $\gamma$ sub-blocks). In the case where $k\nmid N$ or $\gamma\nmid T$, we can pad the blocks of $\widehat{\Ab^{-1}}$ so that these assumptions are met.

A limitation of our proposed CMIM, is the fact that each worker needs to have full knowledge of $\Ab$, in order to estimate columns of $\Ab^{-1}$ through a least squares solver. The sensitivity of Gaussian elimination and matrix inversion also require that all clients have knowledge of each others' data \cite{CPH22a}. This limitation is shared by other coded federated learning methods, \textit{e.g}. \sfsty{CodedPaddedFL} \cite{SKRA21}, and further justifies our requirement that task allocations need to be carefully distributed across the workers, especially in the context of FL. In contrast to CC and GC; where a master server has access to all the data, in FL the data is inherently distributed across devices, thus GC cannot  be  applied directly. We also assume that the coordinator does not intercept the communication between the clients, otherwise she could recover the local data. Also, we trust that the coordinator will not invert $\widehat{\Ab^{-1}}$, to approximate $\Ab$ --- this would be computationally difficult, for $N$ large.

Before broadcasting the data amongst themselves, the clients encode their block $\Ab_i$, which guarantees security from outside eavesdroppers. When the clients receive the encoded data, they can decrypt and recover $\Ab$. Then, their servers act as the workers of the proposed CMIM and carry out their assigned computations, and directly communicate their computations back to the coordinator. Once the \textit{recovery threshold} (the minimum number of responses needed to recover the computation) is met, the approximation $\widehat{\Ab^{-1}}$ is recoverable.

\subsection{Coded Federated Learning}

There are few works that leverage CC to devise secure FL schemes. Most of these works have focused on distributed regression and iterative methods, which is the primary application for FL \cite{DPYTH19,PDAYTAH20,SKRA21,SRRA22,XARWZ22}. Below, we describe and compare these approaches to our work.

The authors of \cite{DPYTH19} proposed \textit{coded federated learning}, in which they utilize a CMM scheme. Their security relies on the use of random linear codes, to define the parity data. Computations are carried out locally on the systematic data, and only the parity data is sent to the coordinator. The main drawback compared to our scheme is that each worker has to generate a random encoding matrix and apply a matrix multiplication for the encodings, while we use the same $\brs$ generator matrix across the network, based on GC, to linearly encode the local computations. The drawback in our case, is that the workers need to securely share their data with each other. This is an artifact of the operation (inversion) we are approximating, and is inevitable in the general case where $\Ab$ has no structure. Under the relaxed FL setting we are considering, where the data is gathered or generated locally and is not i.i.d., we cannot make any assumptions on the structure of $\Ab$.

In \cite{SKRA21}, two methods were proposed. \sfsty{CodedPaddedFL} combines one-time-padding with GC to carry out the FL task. Some of its disadvantages are that a \textit{one-time-pad} (OTP) needs to be generated by each worker, and that the OTPs are shared with the coordinator, which means that if she gets hold of the encrypted data, she can decrypt it, compromising security. Furthermore, there is a heavy communication load and the coordinator needs to store all the pads in order to recover the computed gradients. In the proposed CMIM, the coordinator generates a set of interpolation points, and shares them with the clients. If the coordinator can intercept the communication between the workers, she can decrypt the encrypted data blocks. The second method proposed in \cite{SKRA21}, \sfsty{CodedSecAgg}, relies on \textit{Shamir's secret sharing} (SSS); which is based on polynomial interpolation over finite fields. In contrast, our CMIM relies on GC and Lagrange interpolation.

Lastly, we discuss the method proposed in \cite{XARWZ22}, which is based on the McEliece cryptosystem, and moderate-density parity-check codes. This scheme considers a communication delay model which defines stragglers as the workers who respond slower than the fastest worker, and time out after a predetermined amount of time $\Delta$. As the iterative SD process carries on, such workers are continuously disregarded. Due to this, there is a data sharing step at each iteration, at which the new stragglers communicate encrypted versions of their data to the active workers. Our scheme is non-iterative, and has a fixed recovery threshold. Unlike some of the works previously mentioned, which guarantee information-theoretic security, the McEliece based systems and our approach have \textit{computational} security guarantees.

\vspace{-1mm}
\subsection{Lagrange Interpolation and Polynomial CCMs}

While there is extensive literature on matrix-vector and  matrix-matrix multiplication, and computing the gradient in the presence of stragglers, there is limited work on computing or approximating the inverse of a matrix \cite{YGK17}. The non-linearity of matrix inversion prohibits linear or polynomial encoding of the data before the computations are to be performed. Consequently, most CCMs cannot be directly utilized. Gradient coding is the appropriate CC set up to consider \cite{TLDK17}, precisely because the encoding takes place once the computation has been completed, in contrast to most CMM methods where the encoding is done by the master, before the data is distributed. This helps improve the recovery threshold, which is a primary objective of the CMM problem.

Here, we give a brief overview of the GC scheme on which our CMIM is based. We also review ``Lagrange Coded Computing'' (LCC), which has relations to our approach. Then, we give a summary of our proposed CMIM. All these rely on Lagrange interpolation over finite fields. We then mention related CMM schemes based on Lagrange or polynomial interpolation.

Gradient codes are a class of codes designed to mitigate the effect of stragglers in data centers, by recovering the gradient of differentiable and additively separable objective functions in distributed first order methods \cite{TLDK17}. The proposed CMIM utilizes $\brs$ generator matrices constructed for GC \cite{HASH17}. The main difference from our work is that in GC the objective is to construct an encoding matrix $\Gb\in\C^{n\times k}$ and decoding vectors $\ab_\I\in\C^{k}$, such that $\ab_{\I}^\top\Gb_{\I}=\vec{\bold{1}}$ for any set of non-straggling workers indexed by $\I$. To do so, the decomposition of the $\brs$ generator matrices $\Gb_{\I}=\Hb_{\I}\Pb$ is exploited, where $\Hb_{\I}$ is a Vandermonde matrix; and the first row of $\Pb$ is equal to $\vec{\bold{1}}$. Subsequently $\ab_{\I}^\top$ is extracted as the first row of $\Hb_{\I}^{-1}$.

In the proposed CMIM framework, the objective is to design an \textit{encoding-decoding pair} $(\Gbt,\Dbt_\I)$ for which $\Dbt_\I\Gbt_\I=\Ib_N$, for all $\I\subsetneq\N_n$ of size $k$. The essential reason for requiring this condition, as opposed to that of GC, is that the empirical gradient of a given dataset is the sum of each individual gradients, while in our scenario if the columns of $\widehat{\Ab^{-1}}$ are summed; they cannot then be recovered.

The state-of-the art CC framework is LCC, which is used to compute arbitrary multivariate polynomials of a given dataset \cite{YSRKSA18}; and has since been considered in various settings \cite{SHPN20,FC19b,SMA21,SLMA21,ZL22}. This approach is based on Lagrange interpolation, and it achieves the optimal trade-off between resiliency, security, and privacy. The problem we are considering is not a multivariate polynomial in terms of $\Ab$. To securely communicate $\Ab$ to the workers, we encode it through Lagrange interpolation. Though similar ideas appear in LCC, the purpose and application of the interpolation is different. Furthermore, LCC is a \textit{point-based} approach and requires additional interpolation and linear combination steps after the decoding takes place, while ours is a \textit{coefficient-based} CCM \cite{KD22}.

Recall that the workers in the CMIM must compute blocks of $\widehat{\Ab^{-1}}$. Once they complete their computations, they encode them by computing a linear combination with coefficients determined by a sparsest-balanced MDS generator matrix. Referring to the advantages claimed for CMIM in Section \ref{intro}, working with MDS generator matrices allows us to meet points (i) and (ii), while $\brs$ generator matrices also helps us satisfy (iii). Once the recovery threshold is met, the coordinator can recover the approximation $\widehat{\Ab^{-1}}$. The structure of sparsest-balanced generator matrices is also leveraged to optimally allocate tasks to the workers, while linear encoding is what allows minimal communication load from the workers to the master. Security against eavesdroppers is guaranteed by encoding the local data through a modified Lagrange interpolation polynomial, before it is shared by the clients. This CMIM also extends to approximating $\Ab^{\dagger}$ \cite{CPH22a}.

Some of the earliest interpolation based CMM schemes are the Polynomial \cite{YMAA17} and \sfsty{MatDot} codes \cite{DFHJCG19}, both of which are point-based. The construction of `Polynomial Codes' has since been generalized to `Entangled Polynomial Codes' \cite{YMAA20}, which define similar polynomials to ours \eqref{lagr_pol_matr}, though their use differs. We use \eqref{lagr_pol_matr} to encrypt the clients' data blocks; and our decryption is an evaluation at a finite field point. For Entangled Polynomial Codes two such polynomials are defined; one for each input matrix, and their product determines another degree $R-2$ polynomial which is evaluated by each worker at a different point, before proceeding to the decoding step. 

In \sfsty{MatDot} codes \cite{DFHJCG19} two polynomials are defined, one corresponding to each input, where instead of the Lagrange polynomial in \eqref{lagr_pol_matr}; a monic monomial is multiplied by the partitions of the respective block submatrices. Then, analogous steps to those of Entangled Polynomial Codes take place, in order to recover the matrix product.

\vspace{-1mm}
\section{Preliminary Background}
\label{prel_sec}

The set of $N\times N$ invertible matrices is denoted by $\GL_N(\R)$. Recall that $\Ab\in\GL_N(\R)$ has a unique inverse $\Ab^{-1}$, such that $\Ab\Ab^{-1}=\Ab^{-1}\Ab=\Ib_N$. The simplest way of computing $\Ab^{-1}$ is by performing Gaussian elimination on $\big[\Ab|\Ib_N\big]$, which gives $\big[\Ib_N\big|\Ab^{-1}]$ in $\ow(N^3)$ operations. In Algorithm \ref{inv_alg}, we approximate $\Ab^{-1}$ column-by-column. We denote the $i^{th}$ row and column of $\Ab$ respectively by $\Ab_{(i)}$ and $\Ab^{(i)}$. The \textit{condition number} of $\Ab$ is $\kappa_2=\|\Ab\|_2\|\Ab^{-1}\|_2$. The largest, smallest and $i^{th}$ singular values of $\Ab$ are denoted by $\sigma_{\text{min}}(\Ab)$, $\sigma_{\text{max}}(\Ab)$ and $\sigma_i(\Ab)$ respectively. For $\I$ an index subset of the rows of a matrix $\Mb$, the matrix consisting only of the rows indexed by $\I$, is denoted by $\Mb_\I$. We denote the set of integers between $1$ and $\nu$ by $\N_\nu$. The support of a vector $\vb$ is denoted by $\supp(\vb)$, and the number of nonzero elements of $\Ab$ by $\nnz(\Ab)$.

In the proposed algorithm we approximate $N$ instances of the least squares minimization problem
\begin{equation}
\label{OLS}  
  \theta^{\star}_{ls} = \arg\min_{\theta\in\R^M} \left\{\|\Ab\theta-\yb\|_2^2\right\} 
\end{equation}
for $\Ab\in\R^{N\times M}$ and $\yb\in\R^N$. In many applications $N\gg M$, where the rows represent the feature vectors of a dataset. This has the closed-form solution $\theta^{\star}_{ls} = \Ab^{\dagger}\yb$.

Computing $\Ab^{\dagger}$ to solve \eqref{OLS} is intractable for large $M$, as it requires computing the inverse of $\Ab^\top\Ab$. Instead, we use gradient methods to get \textit{approximate} solutions, \textit{e.g}. SD or CG, which require less operations, and can be done distributively. One could use second-order methods; \textit{e.g}. Newton–Raphson, Gauss-Newton, Quasi-Newton, BFGS, or Krylov subspace methods instead. This would be worthwhile future work.

When considering a minimization problem with a convex differentiable objective function $\psi\colon\Theta\to\R$ over an open convex set $\Theta\subseteq \R^M$, {as in \eqref{OLS}}, the SD procedure selects an initial $\theta^{[0]}\in\Theta$, and then updates $\theta$ according to:
$$ \theta^{[t+1]}=\theta^{[t]}-\xi_t\cdot\nabla_{\theta}\psi(\theta^{[t]}), \quad \text{ for } t=0,1,2,\ldots $$
until a termination criterion is met, for $\xi_t$ the step-size. The CG method is the most used and prominent iterative procedure for numerically solving systems of positive-definite equations.

Our proposed coding scheme is defined over the multiplicative cyclic group $(\F_q^\times,\cdot)$, for $\F_q$ the finite field of $q$ elements and $\F_q^{\times}=\F_q\backslash\{0_{\F_q}\}$ its set of units. For implementation purposes, we identify $\F_q^\times$ with its realization in $\C$ as a subgroup of the circle group, since we assume our data is over $\R$. All operations can therefore be carried out over $\C$. Specifically, we identify $\beta$ as an arbitrary primitive generator of $(\F_q^\times,\cdot)$. One such case is to identify $\beta\mapsto e^{2\pi i/(q-1)}$. Thus, for all $j\in\N_{q-1}$; we identify $\beta^j\mapsto e^{2\pi ij/(q-1)}$.

\vspace{-1mm}
\section{Balanced Reed-Solomon Codes}
\label{BRS_sec}

A Reed-Solomon code $\rs_q[n,k]$ over $\F_q$ for $q>n>k$, is the encoding of polynomials of degree at most $k-1$, for $k$ the message length and $n$ the code length. It represents our message over the \textit{defining set of points} $\Apzc=\{\alpha_j\}_{j=1}^n\subset\F_q$
\begin{align*}
  \rs_q[n,k]=\Big\{\big[f&(\alpha_1),f(\alpha_2),\cdots,f(\alpha_n)\big] \ \Big|\\
  & f(X)\in\F_q[X] \text{ of degree }\leqslant k-1 \Big\}
\end{align*}
where $\alpha_j=\alpha^j$, for $\alpha$ a primitive root of $\F_q$. Hence, each $\alpha_i$ is distinct. A natural interpretation of $\rs_q[n,k]$ is through its encoding map. Each message $(m_0,...,m_{k-1})\in\F_q^k$ is interpreted as $f(\x)=\sum_{i=0}^{k-1}m_i\x^i\in\F_q[\x]$, and $f$ is evaluated at each point of $\Apzc$. From this, $\rs_q[n,k]$ can be defined through the generator matrix
\vspace{-1mm}
$$ \Gb = \begin{pmatrix} 1 & \alpha_1 & \alpha_1^2 & \hdots & \alpha_1^{k-1} \\ 1 & \alpha_2 & \alpha_2^2 & \hdots & \alpha_2^{k-1} \\ \vdots & \vdots & \vdots & \ddots & \vdots \\ 1 & \alpha_n & \alpha_n^2 & \hdots & \alpha_n^{k-1} \end{pmatrix} \in \F_q^{n\times k} , $$
thus, $\rs$ codes are linear codes over $\F_q$. Furthermore, they attain the Singleton bound, \textit{i.e.} $d=n-k+1$, where $d$ is the code's distance, which implies that they are MDS.

Balanced Reed-Solomon codes \cite{HLH16,HLH16b} are a family of linear MDS error-correcting codes with generator matrices $\Gb\in\F_q^{n\times k}$ that are:
\begin{itemize}
  \item \textbf{sparsest}: each \textit{column} has the least possible number of nonzero entries
  \item \textbf{balanced}: each \textit{row} contains the same number of nonzero entries
\end{itemize}
for the given code parameters $k$ and $n$. The design of these generators are suitable for our purposes, as:
\begin{enumerate}
  \item we have balanced loads across homogeneous workers,
  \item sparse generator matrices reduce the computation tasks across the network,
  \item the MDS property permits an efficient decoding step,
  \item linear codes produce a compressed representation of the encoded blocks.
\end{enumerate}

\vspace{-1mm}
\subsection{Balanced Reed-Solomon Codes for CC}
\label{BRS_CC_subs}

\vspace{-1mm}
In the proposed CMIM, we leverage $\brs$ generator matrices to approximate $\Ab^{-1}$ distributively. For simplicity, we will consider the case where $d=s+1=\frac{nw}{k}$ is a positive integer\footnote{The case where $\frac{nw}{k}\in\Q_+\backslash\Z_+$ is analyzed in \cite{HASH17}, and also applies to our approach. We restrict our discussion to the case where $\frac{nw}{k}\in\Z_+$.}, for $n$ the number of workers and $s$ the number of stragglers. Furthermore, $d$ is the distance of the code and $\|\Gb^{(j)}\|_0=d$ for all $j\in\N_k$; $\|\Gb_{(i)}\|_0=w$ for all $i\in\N_n$, and $d>w$ since $n>k$. For decoding purposes, we require that at least $k=n-s$ workers respond. Consequently, $d=s+1$ implies that $n-d=k-1$. For simplicity, we also assume $d\geqslant n/2$. In our setting, each column of $\Gb$ corresponds to a computation task of computing a block of $\widehat{\Ab^{-1}}$; which we will denote by $\Acalh_i$, and each row corresponds to a worker.

Our choice of such a generator matrix $\Gb\in\F_q^{n\times k}$, solves
\begin{equation}
\label{min_G_problem}
\begin{aligned}
\arg\min_{\Gb\in\F_q^{n\times k}} \quad & \big\{\nnz(\Gb)\big\}\\
\textrm{s.t.} \quad & \|\Gb_{(i)}\|_0\geqslant w,\ \forall i\in\N_n\\
  & \|\Gb^{(j)}\|_0\geqslant d,\ \forall j\in\N_k\\
  & \rank(\Gb_\I)=k,\ \forall \I:|\I|=k
\end{aligned}
\end{equation}
which determines an optimal task allocation among the workers of the proposed CMIM. The first and second constraints are analogous to the bound of \cite[Theorem 1]{TLDK17}, which is met with equality in ``perfectly balanced GC schemes''. This theorem states that if all rows of $\Gb$ have the same number of nonzeros, then $\|\Gb_{(i)}\|_0\geqslant k(s+1)/n$, for all $i$. By construction, the generator matrix $\Gb$ we propose, meets the first and second constraints with equality, for all $i\in\N_n$ and $j\in\N_k$.

Under the above assumptions, the entries of the generator matrix of a $\brs_q[n,k]$ code meet the following:
\begin{itemize}
  \item each column is sparsest, with exactly $d$ nonzero entries
  \item each row is balanced, with $w=\frac{dk}{n}$ nonzero entries
\end{itemize}
where $d$ equals to the number of workers who are tasked to compute each block, and $w$ is the number of blocks that are computed by each worker.

Each column $\Gb^{(j)}$ corresponds to a polynomial $p_j(\x)$, whose entries are the evaluation of the polynomial we define in \eqref{lagr_polys} at each of the points of the defining set $\Apzc$, \textit{i.e.} $\Gb_{ij}=p_j(\alpha_i)$ for $(i,j)\in\N_n\times\N_k$. To construct the polynomials $\{p_j(\x)\}_{j=1}^k$, for which $\text{deg}(p_j)\leqslant\nnz(\Gb^{(j)})=n-d=k-1$, we first need to determine a sparsest and balanced \textit{mask matrix} $\Mb\in\{0,1\}^{n\times k}$, which is $\rho$-sparse for $\rho=\frac{d}{n}$; \textit{i.e.} $\nnz(\Gb)=\rho nk$. We use the construction from \cite{HASH17}, though it is fairly easy to construct more general such matrices, by using the Gale-Ryser Theorem \cite{DSDY13,Kra96}. Even though this was not pointed out in \cite{HASH17}, their construction of $\Mb$ (Algorithm \ref{RBMM}) does not always produce a mask matrix of the given parameters when we select $d<n/2$. This is why in our work we require $d\geqslant n/2$. Furthermore, deterministic constructions resemble generator matrices of cyclic codes.

For our purposes we use $\Bcal=\{\beta_j\}_{j=1}^n$ as our defining set of points, where each point corresponds to the worker with the same index. The objective now is to devise the polynomials $p_j(\x)$, for which $p_j(\beta_i)=0$ if and only if $\Mb_{ij}=0$. Therefore:
\begin{enumerate}[label=(\Roman*)]
  \item $\Mb_{ij}=0 \quad \implies \quad (\x-\beta_i)\mid p_j(\x)$
  \item $\Mb_{ij}\neq0 \quad \implies \quad p_j(\beta_i)\in\F_q^{\times}$
\end{enumerate}
for all pairs $(i,j)$.

The construction of $\brs_q[n,k]$ from \cite{HLH16} is based on what the authors called \textit{scaled polynomials}. Below, we summarize the polynomial construction based on Lagrange interpolation \cite{HASH17}. We then prove a simple but important result that allows us to efficiently perform the decoding step.

The univariate polynomials corresponding to each column $\Gb^{(j)}$, are defined as:
\begin{equation}
\label{lagr_polys}
  p_j(\x) \coloneqq \prod\limits_{i:\Mb_{ij}=0}\left(\frac{\x-\beta_i}{\beta_j-\beta_i}\right) = \sum\limits_{\iota=1}^{k}p_{j,\iota}\cdot\x^{\iota-1}\in\F_q[\x]
\end{equation}
which satisfy (I) and (II). By the $\bch$ bound \cite[Chapter 9]{Mc01}, it follows that $\text{deg}(p_j)\geqslant n-d=k-1$ for all $j\in\N_k$. Since each $p_j(\x)$ is the product of $n-d$ monomials, we conclude that the bound on the degree is satisfied and met with equality, hence $p_{j,\iota}\in\F_q^{\times}$ for all coefficients.

By construction, $\Gb$ is decomposable into a Vandermonde matrix $\Hb\in\Bcal^{n\times k}$ and a matrix comprised of the polynomial coefficients $\Pb\in(\F_q^{\times})^{k\times k}$ \cite{HASH17}. Specifically, $\Gb=\Hb\Pb$ where $\Hb_{ij}=\beta_i^{j-1}=\beta^{i(j-1)}$ and $\Pb_{ij}=p_{j,i}$ are the coefficients from \eqref{lagr_polys}. This can be interpreted as $\Pb^{(j)}$ defining the polynomial $p_j(\x)$, and $\Hb_{(i)}$ is comprised of the first $k$ positive powers of $\beta_i$ in ascending order, therefore
\begin{equation*}
  p_j(\beta_i) = \sum\limits_{\iota=1}^{k}p_{j,\iota}\cdot \beta_i^{\iota-1} = \langle\Hb_{(i)},\Pb^{(j)}\rangle .
\end{equation*}

The following lemmas will help us respectively establish in our CC setting the efficiency of our decoding step and the optimality of the allocated tasks to the workers. For Lemma \ref{inverse_lem}, recall that efficient matrix multiplication algorithms have complexity $\ow(N^\omega)$, for $\omega<2.372$ the \textit{matrix multiplication exponent} \cite{WXXZ23}.

\begin{Lemma}
\label{inverse_lem}
The restriction $\Gb_{\I}\in\F_q^{k\times k}$ of $\Gb$ to any of its $k$ rows indexed by $\I\subsetneq\N_n$, is an invertible matrix. Moreover, its inverse can be computed online in $\ow(k^2+k^{\omega})$ operations.
\end{Lemma}

\begin{proof}
The matrices $\Hb$ and $\Pb$ are of size $n\times k$ and $k\times k$ respectively. The restricted matrix $\Gb_{\I}$ is then equal to $\Hb_{\I}\Pb$, where $\Hb_{\I}\in\F_q^{k\times k}$ is a square Vandermonde matrix, which is invertible in $\ow(k^2)$ time \cite{BP70}. Specifically
$$ \Hb_{\I} = \begin{pmatrix} 1 & \beta_{\I_1} & \beta_{\I_1}^2 & \hdots & \beta_{\I_1}^{k-1} \\ 1 & \beta_{\I_2} & \beta_{\I_2}^2 & \hdots & \beta_{\I_2}^{k-1} \\ \vdots & \vdots & \vdots & \ddots & \vdots \\ 1 & \beta_{\I_k} & \beta_{\I_k}^2 & \hdots & \beta_{\I_k}^{k-1} \end{pmatrix} \in \F_q^{k\times k} . $$
It follows that
$$ \det(\Hb_{\I}) = \prod\limits_{\{i<j\}\subseteq\I}(\beta_j-\beta_i) $$
which is nonzero, since $\beta$ is primitive. Therefore, $\Hb_\I$ is invertible. By \cite[Lemma 1]{HLH16} and the $\bch$ bound, we conclude that $\Pb$ is also invertible. Hence, $\Gb_\I$ is invertible for any set $\I$.

Note that the inverse of $\Pb$ can be computed a priori by the master before we deploy our CCM. Therefore, computing $\Gb_\I^{-1}$ online with knowledge of $\Pb^{-1}$, requires an inversion of $\Hb_\I$ which takes $\ow(k^2)$ operations; and then multiplying it by $\Pb^{-1}$. Thus, it requires $\ow(k^2+k^{\omega})$ operations in total.
\end{proof}

\begin{Lemma}
\label{lem_sol_opt_prob}
The generator matrix $\Gb\in\F_q^{n\times k}$ of a $\brs_q[n,k]$ MDS code defined by the polynomials $p_j(\x)$ of \eqref{lagr_polys}, solves the optimization problem \eqref{min_G_problem}.
\end{Lemma}

\begin{proof}
The first two constraints are satisfied by the construction of $\Gb$, which meets the sparsest and balanced constraints with equality; for the given parameters. The last constraint is implied by the MDS theorem, which states that every set of $k$ rows of $\Gb$ is linearly independent.

The sparsity constraints of \eqref{min_G_problem} imply that $\nnz(\Gb)\geqslant\max\{nw,kd\}$, and for our parameters we have $nw=kd$. Both the first and second constraints are met with equality for the chosen $\Gb$. Moreover
\begin{align*}
  \nnz(\Gb) &=\sum_{j\in\N_k}\nnz(\Gb^{(j)})\\
  &=\sum_{j\in\N_k}\#\big\{p_j(\beta_i)\neq0 : \beta_i\in\Bcal\big\}\\
  &=\sum_{j\in\N_k}n-\big\{i : \Mb_{ij}=0\big\}\\
  &=\sum_{j\in\N_k}n-(n-d)\\
  &=kd
\end{align*}
and the proof is complete.
\end{proof}

We conclude this subsection by recalling how the decomposition $\Gb=\Hb\Pb$ is utilized for GC \cite{HASH17}. Each column of $\Gb$ corresponds to a partition of the data whose partial gradient is to be computed. The polynomials are judiciously constructed in this scheme, such that the constant term of each polynomial is 1, thus $\Pb_{(1)}=\vec{\bold{1}}$. By this, the decoding vector $\ab_{\I}^\top$ is the first row of $\Hb_{\I}^{-1}$, for which $\ab_{\I}^\top\Hb_{\I}=\bold{e}_1^\top$. A direct consequence of this is that $\ab_{\I}^\top\Gb_{\I}=\bold{e}_1^\top\Pb=\Pb_{(1)}=\vec{\bold{1}}$, which is the objective for constructing a GC scheme.

\vspace{-1mm}
\section{Inverse Approximation Algorithm}
\label{Appr_alg_sec}

Our goal is to estimate $\Ab^{-1}=\big[\bb_1 \ \cdots \ \bb_N \big]$, for $\Ab$ a square matrix of order $N$. A key property to note is
$$ \Ab\Ab^{-1} = \Ab\big[\bb_1 \ \cdots \ \bb_N \big] = \big[\Ab\bb_1 \ \cdots \ \Ab\bb_N \big] = \bold{I}_N $$
which implies that $\Ab\bb_i=\eb_i$ for all $i\in\N_N$, where $\eb_i$ are the standard basis column vectors. Assume for now that we use any black-box least squares solver to estimate
\begin{equation}
\label{inv_LS}
  \hat{\bb}_i \approx \arg\min_{\bb\in\R^N} \Big{\{g_i(\bb)\coloneqq\|\Ab\bb-\eb_i\|_2^2}\Big\}
\end{equation}
which we call $N$ times, to recover $\widehat{\Ab^{-1}} \coloneqq \big[ \hat{\bb}_1 \ \cdots \ \hat{\bb}_N \big]$. This approach may be viewed as approximating
\begin{equation*}
\label{inv_LS_F}
  \widehat{\Ab^{-1}} \approx \arg\min_{\ \ \Bb\in\R^{N\times N}}\left\{\|\Ab\Bb-\bold{I}_N\|_F^2\right\}.
\end{equation*}
Alternatively, one could estimate the rows of $\Ab^{-1}$. Algorithm \ref{inv_alg} shows how this can be performed by a single server.

\begin{algorithm}[h]
\label{inv_alg}
\SetAlgoLined
\KwIn{$\Ab\in\GL_N(\R)$}
  \For{i=1 to N}
  {
    approximate $\hat{\bb}_i \approx \arg\min_{\bb\in\R^N} \left\{\|\Ab\bb-\eb_i\|_2^2\right\}$ 
  }
 \Return $\widehat{\Ab^{-1}} \gets \big[ \hat{\bb}_1 \ \cdots \ \hat{\bb}_N \big]$
 \caption{Estimating $\Ab^{-1}$}
\end{algorithm}

In the case where SD is used to approximate $\hat{\bb}_i$ from \eqref{inv_LS}, the overall operation count is $\ow(\T_i N^2)$; for $\T_i$ the total number of descent iterations used. An upper bound on the number of iterations can be determined by the underlying termination criterion, \textit{e.g}. the criterion $g_i(\bbh^{[t]})-g_i(\bb^{\star}_{ls})\leqslant\epsilon$ is guaranteed to be satisfied after $\T=\ow(\log(1/\epsilon))$ iterations \cite{BV04}. The overall error of $\widehat{\Ab^{-1}}$ may be quantified as
\begin{itemize}
  \item $\text{err}_{\ell_2}(\widehat{\Ab^{-1}}) \coloneqq \|\widehat{\Ab^{-1}}-\Ab^{-1}\|_2$
  \item $\text{err}_F(\widehat{\Ab^{-1}}) \coloneqq \|\widehat{\Ab^{-1}}-\Ab^{-1}\|_F$
  \item $\text{err}_{\text{r}F}(\widehat{\Ab^{-1}}) \coloneqq \frac{\|\widehat{\Ab^{-1}}-\Ab^{-1}\|_F}{\|\Ab^{-1}\|_F} = \frac{\left(\sum\limits_{i=1}^{N} \|\Ab\hat{\bb}_i-\eb_i\|_2^2\right)^{1/2}}{\|\Ab^{-1}\|_F}$.
\end{itemize}

To approximate $\Ab^{-1}$ distributively, each of the $n$ workers are asked to estimate $\tau$-many $\bbh_i$'s in parallel. When using SD, the worst-case runtime by the workers is $\ow(\tau\T_{\text{max}}N^2)$, for $\T_{\text{max}}$ the maximum number of iterations of SD among the workers. If CG is used, each worker needs no more than a total of $N\tau$ CG steps to exactly compute its task, \textit{i.e.} $\ow(\tau N\kappa_2)$ operations; as each instance of \eqref{inv_LS} is expected to converge in $\ow(\kappa_2)$ iterations, which is the worst case runtime \cite{She94,TB97}.

In order to bound $\text{err}_{\text{r}F}(\widehat{\Ab^{-1}})=\frac{\|\widehat{\Ab^{-1}}-\Ab^{-1}\|_F}{\|\Ab^{-1}\|_F}$, we first upper bound the numerator and then lower bound the denominator. Since $\|\Ab^{-1}-\widehat{\Ab^{-1}}\|_F^2=\sum_{i=1}^N\|\Ab^{-1}\eb_i-\hat{\bb}_i\|_2^2$, bounding the numerator reduces to bounding $\|\Ab^{-1}\eb_i-\hat{\bb}_i\|_2^2$ for all $i\in\N_N$. This is straightforward
\begin{align}
\label{upp_bd_numer}
  \|\Ab^{-1}\eb_i-\hat{\bb}_i\|_2^2 &\overset{\Diamond}{\leqslant} 2\left(\|\Ab^{-1}\eb_i\|_2^2+\|\hat{\bb}_i\|_2^2\right)\notag\\
  &\overset{\$}{\leqslant} 2\left(\|\Ab^{-1}\|_2^2\cdot\|\eb_i\|_2^2+\|\hat{\bb}_i\|_2^2\right)\notag\\
  &= 2\left(1/\sigma_{\text{min}}(\Ab)^2+\|\hat{\bb}_i\|_2^2\right)
\end{align}
where in $\Diamond$ we use the fact that $\|\ub-\vb\|_2^2\leqslant2(\|\ub\|_2^2+\|\vb\|_2^2)$, and in $\$$ the submultiplicativity of the $\ell_2$-norm is invoked. For the denominator, by the definition of the Frobenius norm
\begin{equation}
\label{low_bd_denom}
  \|\Ab^{-1}\|_F^2=\sum_{i=1}^N\frac{1}{\sigma_i(\Ab)^2} \geqslant \frac{N}{\sigma_{\text{max}}(\Ab)^2}\ .
\end{equation}
By combining \eqref{upp_bd_numer} and \eqref{low_bd_denom} we get
\begin{align*}
  \text{err}_{\text{r}F}(\widehat{\Ab^{-1}}) &\leqslant \sqrt{2}\left(\frac{N/\sigma_{\text{min}}(\Ab)^2+\sum_{i=1}^N\|\hat{\bb}_i\|_2^2}{N/\sigma_{\text{max}}(\Ab)^2}\right)^{1/2}\\
  &= \sqrt{2}\left(\kappa_2^2+\frac{\sigma_{\text{max}}(\Ab)^2}{N}\cdot\sum_{i=1}^N\|\hat{\bb}_i\|_2^2\right)^{1/2} .
\end{align*}

This is an additive error bound in terms of the problem's condition number, which also shows a dependency on the estimates $\{\bbh_i\}_{i=1}^N$. Propositions \ref{rel_err_bound_prop} and \ref{rel_err_bound_prop_CG} give error bounds when using SD and CG as the subroutine of Algorithm \ref{inv_alg} respectively.

\begin{Prop}
\label{rel_err_bound_prop}
  For $\Ab\in\GL_N(\R)$, we have $\err_{F}(\widehat{\Ab^{-1}})\leqslant \frac{\epsilon\sqrt{N/2}}{\sigma_{\text{min}}(\Ab)^2}$ and $\err_{\mathrm{r}F}(\widehat{\Ab^{-1}})\leqslant \frac{\epsilon\sqrt{N/2}}{\sigma_{\mathrm{min}}(\Ab)}$, when using SD to solve \eqref{inv_LS} with termination criteria $\|\nabla g_i(\bb^{[t]})\|_2\leqslant\epsilon$ for each $i$.
\end{Prop}

\begin{proof}
Recall that for a strongly-convex function with strong-convexity parameter $\mu$, we have the following optimization gap \cite[Section 9.1.2]{BV04}
\begin{equation}
\label{opt_gap}
  g_i(\bb)-g_i(\bb_{ls}^{\star})\leqslant\frac{1}{2\mu}\cdot\|\nabla g_i(\bb)\|_2^2 \ .
\end{equation}
For $\Ab\in\GL_N(\R)$ in \eqref{inv_LS}, the constant is $\mu=\sigma_{\text{min}}(\Ab)^2$. By fixing $\epsilon=\sqrt{2\sigma_{\text{min}}(\Ab)^2\eta}$, we have $\eta=\frac{1}{2}\cdot\left(\frac{\epsilon}{\sigma_{\text{min}}(\Ab)}\right)^2$. Thus, by \eqref{opt_gap} and our termination criterion:
$$ \|\nabla g_i(\bb)\|_2\leqslant\sqrt{2\sigma_{\text{min}}(\Ab)^2\eta} \quad \implies \quad g_i(\bb)-g_i(\bb_{ls}^{\star})\leqslant\eta\ , $$
so when solving \eqref{inv_LS} we get
$$ g_i(\bb)-g_i(\bb_{ls}^{\star}) = g_i(\bb)-0 = \|\Ab\hat{\bb}_i-\eb_i\|_2^2\ , $$
hence
\begin{equation}
\label{summand_obj_bound}
  \|\Ab\hat{\bb}_i-\eb_i\|_2^2 \leqslant \frac{1}{2}\cdot\left(\frac{\epsilon}{\sigma_{\text{min}}(\Ab)}\right)^2
\end{equation}
for all $i\in\N_N$. We want an upper bound for each summand $\|\Ab^{-1}\eb_i-\hat{\bb}_i\|_2^2$ of the numerator of $\text{err}_{\text{r}F}(\widehat{\Ab^{-1}})^2$:
\begin{align}
  \|\Ab^{-1}\eb_i-\hat{\bb}_i\|_2^2 &= \|\Ab^{-1}(\eb_i-\Ab\hat{\bb}_i)\|_2^2\notag\\
  &\leqslant \|\Ab^{-1}\|_2^2\cdot\|\eb_i-\Ab\hat{\bb}_i\|_2^2\notag\\
  &\overset{\sharp}{\leqslant} \|\Ab^{-1}\|_2^2\cdot\frac{1}{2} \cdot \left(\frac{\epsilon}{\sigma_{\text{min}}(\Ab)}\right)^2\label{deriv_rF_bd_1}\\
  &=\frac{\epsilon^2}{2\sigma_{\text{min}}(\Ab)^4}\label{deriv_rF_bd}
\end{align}
where $\sharp$ follows from \eqref{summand_obj_bound}, thus $\text{err}_{F}(\widehat{\Ab^{-1}})^2\leqslant\frac{N\epsilon^2}{2\sigma_{\text{min}}(\Ab)^4}$. Substituting \eqref{deriv_rF_bd_1} into the definition of $\text{err}_{\text{r}F}(\widehat{\Ab^{-1}})$ gives us
$$ \text{err}_{\text{r}F}(\widehat{\Ab^{-1}})^2 \leqslant \frac{\|\Ab^{-1}\|_2^2}{\|\Ab^{-1}\|_F^2} \cdot\frac{N}{2} \cdot \left(\frac{\epsilon}{\sigma_{\text{min}}(\Ab)}\right)^2 \overset{\ddagger}{\leqslant}\frac{N\epsilon^2/2}{\sigma_{\mathrm{min}}(\Ab)^2} $$
where $\ddagger$ follows from the fact that $\|\Ab^{-1}\|_2^2\leqslant\|\Ab^{-1}\|_F^2$.
\end{proof}

In the experiments of Subsection \ref{exper_sec}, we verify that Proposition \ref{rel_err_bound_prop} holds for Gaussian random matrices. The dependence on $1/\sigma_{\text{min}}(\Ab)$ is an artifact of using gradient methods to solve the underlying problems \eqref{inv_LS}, since the error will be multiplied by $\|\Ab^{-1}\|_2^2$. In theory, this can be annihilated if one runs the algorithm on $p\Ab$ for $p\approx1/\sigma_{\text{min}}(\Ab)$, followed by multiplication of the final result by $p$. This is a way of preconditioning SD. In practice, the scalar $p$ should not be selected to be much larger than $1/\sigma_{\text{min}}(\Ab)$, as it could result in $\widehat{\Ab^{-1}}\approx\bold{0}_{N\times N}$.

\begin{Prop}
\label{rel_err_bound_prop_CG}
  Assume Algorithm \ref{inv_alg} uses CG to solve \eqref{inv_LS}. Then, in $\ow\left(N\sqrt{\kappa_2}\ln(1/\epsilon)\right)$ iterations, we have $\err_{F}(\widehat{\Ab^{-1}})\leqslant N\epsilon$. Moreover, if $\Ab^\top\Ab$ has $\tilde{N}$ distinct eigenvalues, it converges in at most $\tilde{N}N$ steps.
\end{Prop}

\begin{proof}
By \cite[Section 10]{She94} and \cite[Section 2]{Bub15}, we know that for each subroutine \eqref{inv_LS} of Algorithm \ref{pinv_alg}, CG requires at most $\ow(\sqrt{\kappa_2}\ln(1/\epsilon))$ iterations in order to attain an $\epsilon$-optimal point, for each $\hat{\bb}_i$. Hence, considering all approximate columns $\{\hat{\bb}_i\}_{i=1}^N$, we conclude that the total error in terms of the Frobenius norm of $\widehat{\Ab^{-1}}$, is at most $N\epsilon$.

Recall that in order to solve \eqref{OLS} with the CG method in the case where $\Ab$ is neither symmetric, positive-definite, nor square, we apply the CG iteration to the normal equations: $\Ab^\top\Ab\yb=\Ab^\top\theta$.
This follows by setting the derivative of \eqref{OLS} to zero. In our scenario, we are assuming that $\Ab\in\GL_N(\R)$, hence $\Ab^\top\Ab$ is full-rank and symmetric, thus CG in its simplest form can be used to solve the minimization problems of Algorithm \ref{inv_alg}. By \cite[Theorem 38.4]{TB97}, it follows that each instance of \eqref{inv_LS} converges in at most $\tilde{N}$ steps.
\end{proof}

Even though Proposition \ref{rel_err_bound_prop_CG} guarantees convergence in at most $\tilde{N}N$ steps, it does not assume finite floating-point precision. Therefore, this does not hold in practical settings. Our experiments though show that after significantly less steps, we achieve approximations of negligible error, which is sufficient for ML and FL applications.

\subsection{Numerical Experiments}
\label{exper_sec}

The accuracy of the proposed algorithm was tested on randomly generated matrices, using both SD and CG for the subroutine optimization problems. The depicted results are averages of 20 runs, with termination criteria $\|\nabla g_i(\bb^{[t]})\|_2\leqslant \epsilon$ for SD and $\|\bb_i^{[t]}-\bb_i^{[t-1]}\|_2\leqslant \epsilon$ for CG, for the given $\epsilon$ accuracy parameters. We considered $\Ab\in \R^{100\times 100}$. The error subscripts represent $\mathscr{A}=\{\ell_2,{F},\text{r}F\}$, $\mathscr{N}=\{\ell_2,F\}$. We note that significantly fewer iterations took place when CG was used for the same $\epsilon$, though this depends heavily on the choice of the step-size. The errors observed in the case of CG, are due to floating-point arithmetic. Therefore, as expected; there is a trade-off between accuracy and speed when using SD vs. CG.

\begin{center}
\begin{tabular}{ |p{.5cm}||p{1.05cm}|p{1.1cm}|p{1.1cm}|p{1.1cm}|p{1.1cm}| }
\hline
\multicolumn{6}{|c|}{Average $\widehat{\Ab^{-1}}$ errors, for $\Ab\sim50\cdot \mathcal{N}(0,1)$ --- SD} \\
\hline
$\epsilon$ & $10^{-1}$ & $10^{-2}$ & $10^{-3}$ & $10^{-4}$ & $10^{-5}$ \\
\hline
$\text{err}_{\mathscr{A}}$ & {\small$\ow(10^{-2})$} & {\small$\ow(10^{-5})$} & {\small$\ow(10^{-7})$} & {\small$\ow(10^{-9})$} & {\small$\ow(10^{-12})$} \\
\hline
\end{tabular}
\end{center}

\begin{center}
\begin{tabular}{ |p{.5cm}||p{1.05cm}|p{1.05cm}|p{1.05cm}|p{1.1cm}|p{1.1cm}| }
\hline
\multicolumn{6}{|c|}{{\small Average $\widehat{\Ab^{-1}}$ errors, for $\Ab\sim50\cdot \mathcal{N}(0,1)$ --- CG}} \\
\hline
$\epsilon$ & $10^{-3}$ & $10^{-4}$ & $10^{-5}$ & $10^{-6}$ & $10^{-7}$ \\
\hline
$\text{err}_{\mathscr{N}}$ & {\small$\ow(10^{-3})$} & {\small$\ow(10^{-5})$} & {\small$\ow(10^{-8})$} & {\small$\ow(10^{-11})$} & {\small$\ow(10^{-12})$} \\
$\text{err}_{\text{r}F}$ & {\small$\ow(10^{-3})$} & {\small$\ow(10^{-5})$} & {\small$\ow(10^{-7})$} & {\small$\ow(10^{-10})$} & {\small$\ow(10^{-12})$} \\
\hline
\end{tabular}
\end{center}

We utilized Algorithm \ref{inv_alg} in Newton's method, for classifying images of four and nine from MNIST, by solving a regularized logistic regression minimization problem. For Algorithm \ref{inv_alg}, we used CG with a fixed number of iteration per column estimation. It is clear from Figure \ref{err_MNIST_CG} that we require no more than 18 iterations per column estimate, for $N=785$, to attain the optimal classification rate. With more than 18 CG iterations, the same classification rate was obtained.

\begin{figure}[h]
  \centering
    \includegraphics[scale=.15]{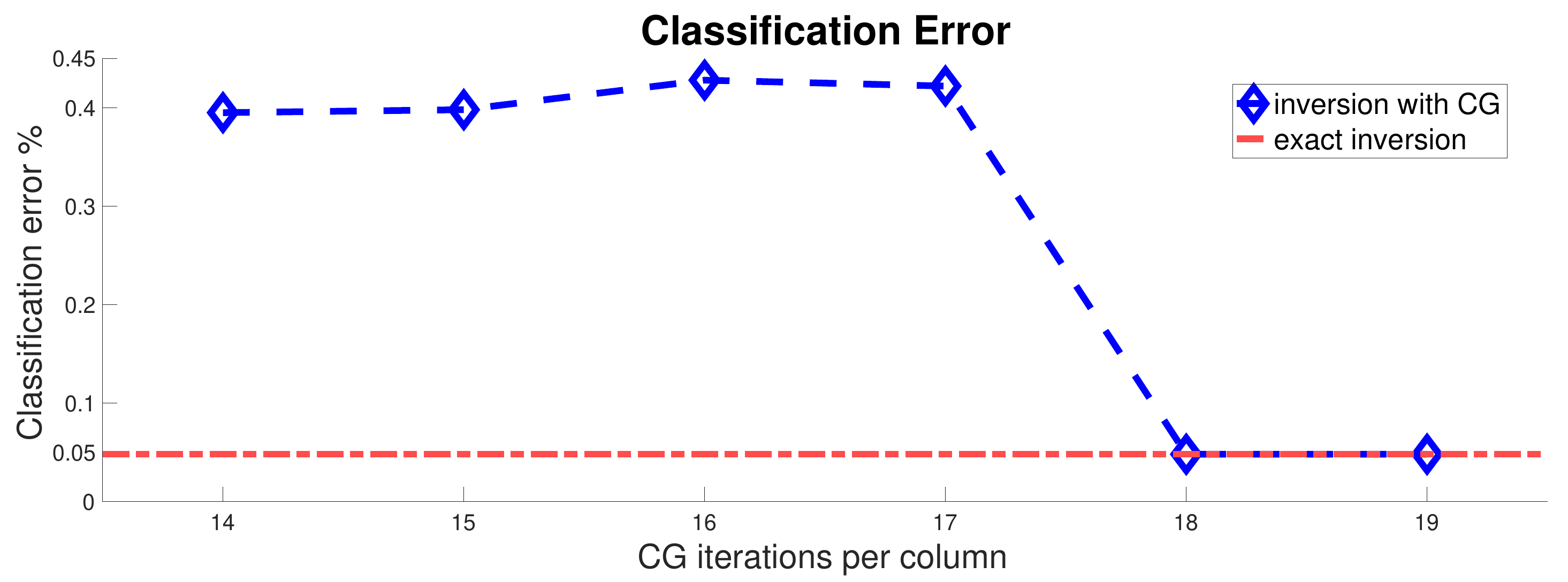}
    \caption{MNIST classification error, where Algorithm \ref{inv_alg} is used in Newton's method. In red, we depict the error when exact inversion was used.}
  \label{err_MNIST_CG}
\end{figure}

\vspace{-1mm}
\section{Secure Coded Matrix Inversion}
\label{FCMI_sec}

In this section, we describe the proposed CMIM (also presented in \cite{CPH22a}) which makes Algorithm \ref{inv_alg} resilient to stragglers, and show how it can be applied to the relaxed FL scenario described in the introduction. The CMIM workflow is depicted in Figure \ref{CMIS_fig}.

Our scheme can be broken up in to four phases: $(\mathrm{a})$ the coordinator shares elements $\beta,\Hh$ of a finite field with all the clients, $(\mathrm{b})$ the clients each generate a \textit{pseudorandom permutation} ($\prp$) $\sigma_\iota$, encrypt their corresponding data block $\Ab_\iota$ through a matrix polynomial $f_\iota(\x)$, and broadcast $\{f_\iota(\x),\sigma_\iota\}$ to the other clients, $(\mathrm{c})$ the clients recover $\Ab$, compute and encode their assigned task $\Wb_\iota$, which is communicated to the coordinator, $(\mathrm{d})$ the coordinator decodes once sufficiently many workers respond. It is also possible that $\beta,\Hh$ are determined collectively by the clients, or by a single client, which makes the data sharing secure against a curious and dishonest coordinator.

\begin{figure}[h]
  \centering
    \includegraphics[scale=.17]{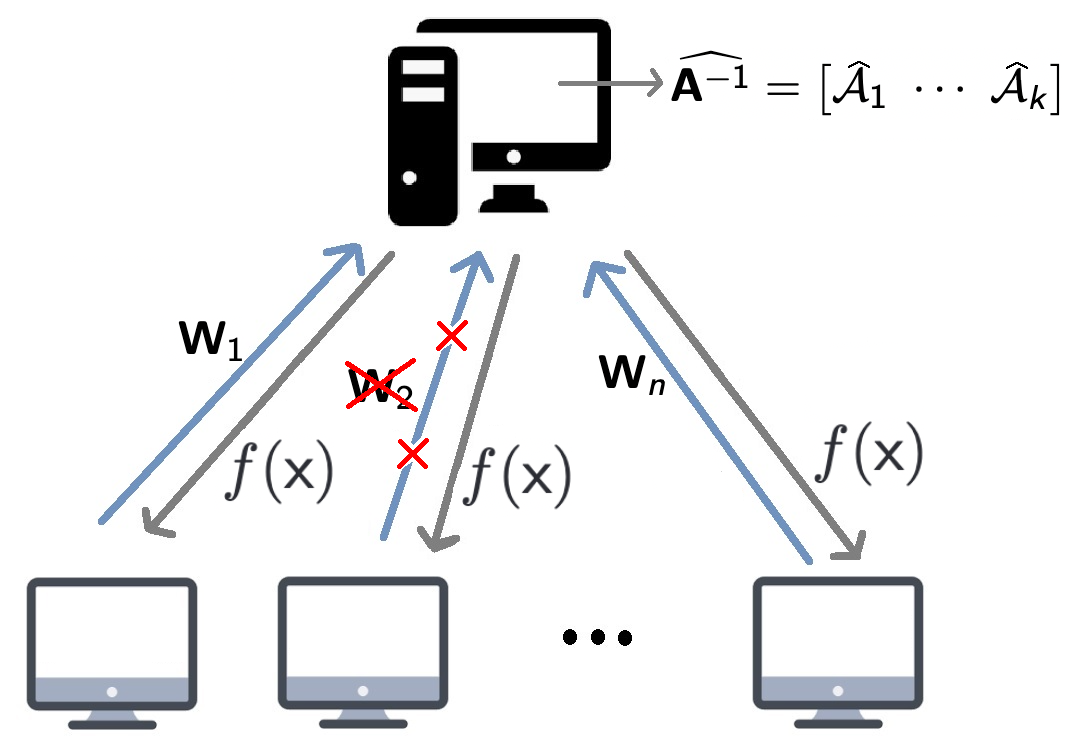}
    \caption{Algorithmic workflow of the CMIM, as proposed in \cite{CPH22a}. The master shares $f(\x)$, an encoding analogous to \eqref{lagr_pol_matr}, along with $\beta,\{\eta_j^{-1}\}_{j=1}^k$. The workers then recover $\Ab$, compute their assigned tasks, and encode them according to $\Gb$. Once $k$ encodings $\Wb_\iota$ are sent back, $\widehat{\Ab^{-1}}$ can be recovered.}
  \label{CMIS_fig}
\end{figure}

In our proposed approach, we assume there is a trustworthy coordinator who shares certain parameters to each of the $k$ clients which constitute the network; \textit{e.g}. hospitals in a health care network, each of which are comprised of multiple servers. What we present works for the case where the clients have local datasets of different sizes, $\{N_i\}_{i=1}^k$. This would result in the encoding functions $f_\iota(\x)$ having different degrees, or their matrix coefficients being of a different size. In our setting we assume the workers are \textit{homogeneous}, \textit{i.e.} they have the same computational power. Therefore, equal computational loads are assigned to each of them. In order to keep the notation and size of the communication loads consistent, we assume w.l.o.g. that $\Ab_\iota\in\R^{N\times T}$ for all $\iota\in\N_k$. If this is not the case, before $f_\iota(\x)$ are determined, the clients could perform a data exchange phase (\textit{e.g}. \cite{XARWZ22}), so that $N_i=N_j$ for all $i\neq j$. By this, it follows that the number of blocks does not have to be equal to the number of clients. The example we describe, is simply a motivation. A flowchart of our approach is presented in Figure \ref{CMI_floatchart}.

Moreover, in the case where $M>N$; for $M=\sum_{i=1}^kN_i$, we can select a subset of features and/or samples, so that the resulting data matrix we consider is square. This can be interpreted as using the surrogate $\Abt=\Sb\Ab$, where $\Sb\in\R^{N\times M}$ is an appropriate (sparse) \textit{sketching} matrix for matrix inversion \cite{Gow16}, which the workers agree on.

\begin{figure}[h]
  \centering
    \includegraphics[scale=.17]{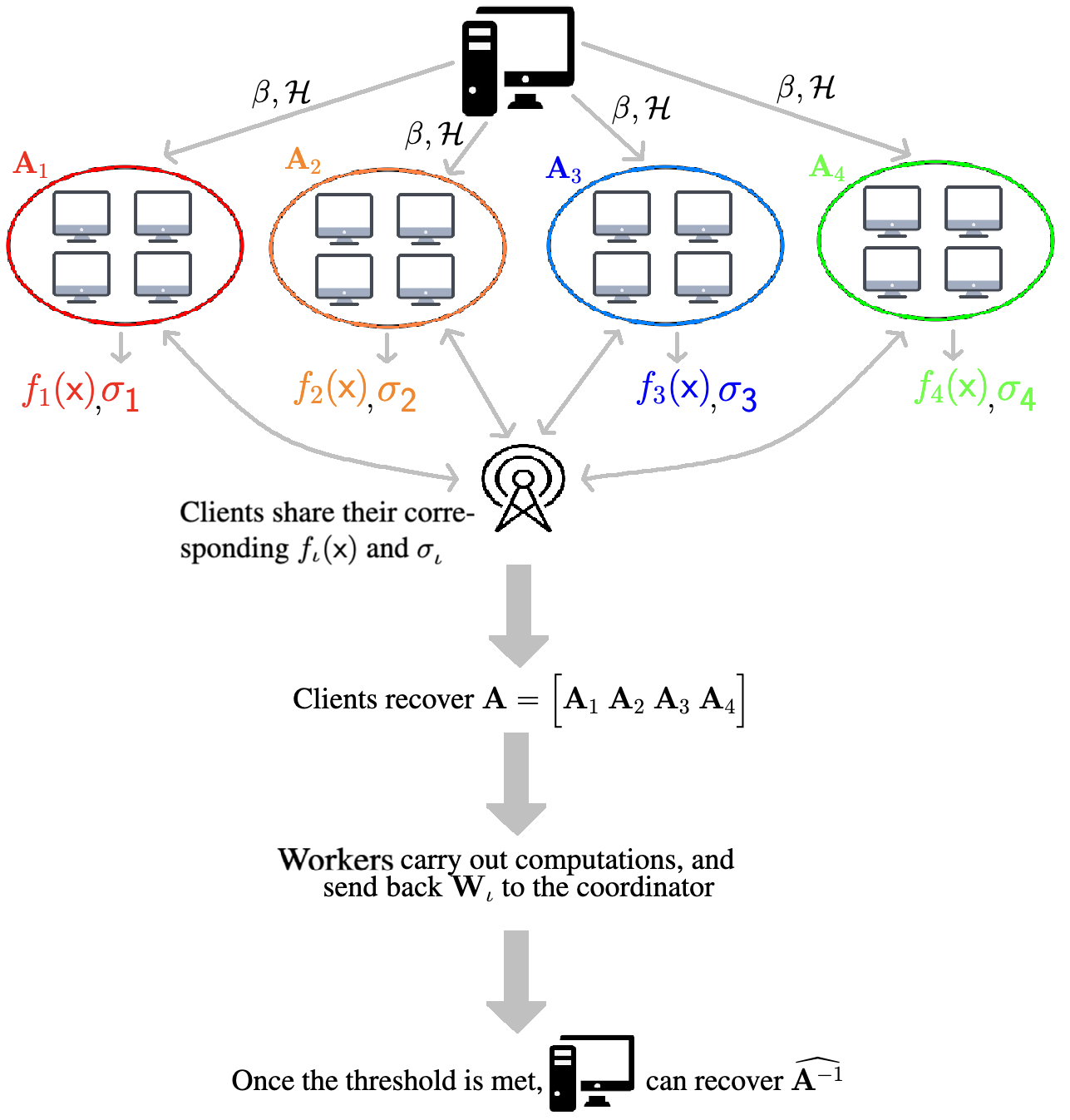}
    \caption{Flowchart of our proposal, where $k=n_i=4$ for all $i\in\N_4$.}
  \label{CMI_floatchart}
\end{figure}

First, in \ref{know_of_A_subs} we argue why all of $\Ab$ needs to be known by each of the workers, in order to recover entries or columns of its inverse. Then, in \ref{enc_A_subs} we focus on phases $(\mathrm{a})$ and $(\mathrm{b})$, where we utilize Lagrange interpolation to securely share $\Ab$ among the workers. We discuss the computation tasks the workers are requested to compute, which are blocks of $\widehat{\Ab^{-1}}$; and collectively correspond to the subroutine problems of Algorithm \ref{inv_alg}. In \ref{comp_subs} we focus on $(\mathrm{c})$ and $(\mathrm{d})$, where we show how the workers encode their computations, and describe the coordinator's decoding step. Optimality of $\brs$ generator matrices in terms of the encoded communication loads is established in \ref{opt_BRS_subs}.

When assuming no floating-point errors, our approach introduces no numerical nor approximation errors. The errors are a consequence of using iterative solvers to estimate \eqref{inv_LS}, which we utilize to linearly separate the computations. Therefore, if the workers can recover the optimal solutions to the underlying minimization problems, our scheme would be \textit{exact}.

\subsection{Knowledge of $\Ab$ is necessary}
\label{know_of_A_subs}

A bottleneck when computing the inverse of a matrix; or estimating its columns, is that the entire matrix needs to be known. A single change in the matrix's entries may result in a non-singular matrix, which conveys how sensitive Gaussian elimination is. Such problems are extensively studied in conditioning and stability of numerical analysis \cite{TB97}, and in perturbation theory. This is not a focus of our work.

In the case where only one column is not known, one can determine the subspace in which the missing column lies, but without the knowledge of at least one entry of that column, it would be impossible to recover that column. Even with such an approach or a matrix completion algorithm, the entire $\Ab$ is determined before we proceed to inverting $\Ab$; or performing linear regression to approximate $\Ab\bb=\eb_i$ as in \eqref{inv_LS}.

A similar issue, relating to our set up, is the case where one of the blocks is different. This could lead to drastic miscalculations. In the following example, we consider $n=k=2$ and $N=4$, where the second worker sends two different blocks, which are indicated by a different color and font:
\begin{equation*}
  {\small \qvec{A}_1 = \begin{pmatrix} 6 & 2 & \textit{\blue2} & \textit{\blue-5}\\ 0 & -1 & \textit{\blue2} & \textit{\blue0}\\ -5 & 6 & \textit{\blue-1} & \textit{\blue-3}\\ 5 & -3 & \textit{\blue-4} & \textit{\blue3} \end{pmatrix} \quad \qvec{A}_2 = \begin{pmatrix} 6 & 2 & \mathfrak{\purple-1} & \mathfrak{\purple-3}\\ 0 & -1 & \mathfrak{\purple5} & \mathfrak{\purple6}\\ -5 & 6 & \mathfrak{\purple3} & \mathfrak{\purple-2}\\ 5 & -3 & \mathfrak{\purple1} & \mathfrak{\purple6} \end{pmatrix}}\ .
\end{equation*}
It follows that $\|\qvec{A}_1^{-1}\|_F\approx90.45$, $\|\qvec{A}_2^{-1}\|_F\approx1$, and $\|\qvec{A}_1^{-1}-\qvec{A}_2^{-1}\|_0=16$; \textit{i.e.} no entries of $\qvec{A}_1^{-1}$ and $\qvec{A}_2^{-1}$ are equal.

Furthermore, by the data processing inequality \cite[Corollary pg. 35]{CT06}, the above imply that no less than $N^2$ information symbols can be known by each worker, while hoping to approximate a column of $\Ab^{-1}$. Hence, all clients need full knowledge of each others information, and cannot communicate less than $NT$ symbols to each other. This is a consequence of the fact that a dense vector is not recoverable from underdetermined linear measurements. They can however send an encoded version of their respective block $\Ab_{\iota}\in\R^{N\times T}$ to the other clients consisting of $NT$ symbols, determined by a modified Lagrange polynomial, which guarantees security against eavesdroppers.

Similar cryptographic protocols date back to the SSS algorithm \cite{Sha79,Bla79}, which is also based on $\rs$ codes. This idea has been extensively exploited in LCC \cite{YSRKSA18}, yet differs from our approach.

\vspace{-2mm}
\subsection{Phases $(\mathrm{a}),(\mathrm{b})$ --- Data Encryption and Sharing}
\label{enc_A_subs}

Let $k,\gamma\in\Z_+$ be factors of $N$ and $T$ respectively, so that $T=\frac{N}{k}$ and $\Gamma=\frac{T}{\gamma}$.\footnote{If $\gamma\nmid T$, append $\bold{0}_{T\times 1}$ to the end of the first $\tilde{\gamma}=T(\bmod\gamma)$ blocks which are each comprised of $\tilde{\Gamma}=\floor{\frac{T}{\gamma}}$ columns of $\Ab_\iota$, while the remaining $\gamma-\tilde{\gamma}$ blocks are comprised of $\tilde{\Gamma}+1$ columns. Now, each block is of size $T\times(\tilde{\Gamma}+1)$.} The coordinator constructs a set of distinct \textit{interpolation points} $\Bcal=\{\beta_j\}_{j=1}^n\subsetneq \F_q^{\times}$, for $q>n\geqslant\gamma$.\footnote{For the encodings of the $\Ab_\iota$'s, $\gamma$ points suffice, and we only need to require $q>\gamma$. We select $\Bcal$ of cardinality $n$ and require $q>n\geqslant\gamma$, in order to reuse $\Bcal$ in our CCM.} To construct this set, it suffices to sample $\beta\in\F_q^{\times}$; any one of the $\phi(q-1)$ primitive roots of $\F_q$ ($\phi$ is Euler's totient function), which is a generator of the multiplicative group $(\F_q^{\times},\cdot)$, and define each point as $\beta_j=\beta^j$. Then, a random multiset $\Hh=\left\{\eta_j\in\F_q^\times\mid\forall j\in\N_\gamma\right\}$ of size $\gamma$ is generated, \textit{i.e.} repetitions in $\Hh$ are allowed, which will be used to remove the structure of the Lagrange coefficients, as the adversaries could exploit their structure to reveal $\beta$.

The element $\beta$ and set $\Hh$, are broadcasted securely to all the workers through a public-key cryptosystem, \textit{e.g}. $\rsa$ or McEliece. Matrices $\Ab_\iota$ are partitioned into $\gamma$ blocks
\begin{equation}
\label{parts_A_iota}
  \Ab_\iota=\Big[\Ab_\iota^{1} \ \cdots \ \Ab_\iota^\gamma\Big] \quad \text{ where } \Ab_\iota^i\in\R^{N\times\Gamma},\ \forall i\in\N_\gamma,
\end{equation}
and each client generates a $\prp$ $\sigma_\iota\in\ S_\gamma$. The blocks $\{\Ab_\iota\}_{\iota=1}^k$ are encrypted locally through the univariate polynomials
\begin{equation}
\label{lagr_pol_matr}
  f_\iota(\x) = \sum\limits_{j=1}^\gamma\Ab_\iota^j\cdot\eta_{\sigma_\iota(j)}\left(\prod\limits_{l\neq j}\frac{\x-\beta_l}{\beta_j-\beta_l}\right)
\end{equation}
for which $f_\iota(\beta_j)=\eta_{\sigma_\iota(j)}\Ab_\iota^j$.

The clients securely broadcast $\{f_\iota(\x),\sigma_\iota\}$ to each other, and their servers can then recover all $\Ab_\iota$'s as follows:
\begin{equation}
\label{decrypt_blocks}
  \Ab_\iota=\Big[\eta_{\sigma_\iota(1)}^{-1}f_\iota(\beta_1) \ \cdots \ \eta_{\sigma_\iota(\gamma)}^{-1}f_\iota(\beta_\gamma)\Big]\in\R^{N\times T}.
\end{equation}
The coefficients of $f_\iota(\x)$ are comprised of $N\Gamma$ symbols, thus, each polynomial consists of a total of $NT$ symbols, which is the minimum number of symbols needed to be communicated. The $\prp$ $\sigma_\iota$ is generated locally by the clients, to ensure that each $f_\iota(\x)$ differs by more than just the matrix partitions.

We assume Kerckhoffs' principle, which states that everyone has knowledge of the system, including the messages $f_\iota(\x)$. For the proposed CMIM, as long as $\{\beta,\Hh\}$ and $\sigma_\iota$ are securely communicated, even if $f_\iota(\x)$ is revealed, the block $\Ab_\iota$ is secure against polynomial-bounded adversaries (this is the security level assumed by the cryptosystems used for the communication).

\begin{Prop}
\label{prop_sec_cryptosystem}
The encryptions of $\Ab_\iota$ through $f_\iota(\x)$, are as secure against eavesdroppers as the public-key cryptosystems which are used when broadcasting $\{\beta,\Hh\}$ and $\sigma_\iota$. To recover $\Ab_\iota$, an adversary needs to intercept both communications, and break both cryptosystems.
\end{Prop}

\begin{proof}
We prove this by contradiction. Assume that an adversary was able to reverse the encoding $f_\iota(\x)$ of $\Ab_\iota$. This implies that he was able to reveal $\beta$ and $\sigma_\iota(\Hh)\coloneqq\{\eta_{\sigma_{\iota(j)}}\}_{j=1}^\gamma$. The only way to reveal these elements, is if he was able to both intercept and decipher the public-key cryptosystem used by the coordinator, which contradicts the security of the cryptosystem.

In order to invert the multiplications of $\sigma_\iota(\Hh)$ for each of the evaluations of $f_\iota(\x)$, both $\Hh$ and $\sigma_\iota$ need to be known. To do so, the adversary needs to intercept both the communication between the coordinator and the clients, and the communication between the clients, as well as breaking both the cryptosystems used to securely carry out these communications.
\end{proof}

\vspace{-3mm}
\subsection{Phases $(\mathrm{c}),(\mathrm{d})$ --- Computations, Encoding and Decoding}
\label{comp_subs}

At this stage, the workers have knowledge of everything they need in order to recover $\Ab$, before they carry out their computation tasks. By \eqref{decrypt_blocks}, the recovery is straightforward.

For Algorithm \ref{inv_alg}, any CCM in which the workers compute an encoding of partitions of the resulting computation $\Eb=\big[E_1 \ \cdots \ E_k \big]$ could be utilized. It is crucial that the encoding takes place on the computed tasks $\{E_i\}_{i=1}^k$ in the scheme, and \textit{not} the assigned data or partitions of the matrices that are being computed over (such CMM leverage the linearity of matrix multiplication), otherwise the algorithm could potentially not return the correct approximation. This also means that utilizing such encryption approaches (\textit{e.g}. \cite{YSRKSA18}) for guaranteeing security against the workers, is not an option. We face these restrictions due to the fact that matrix inversion is a non-linear operator.

The computation tasks $E_i$ correspond to a partitioning $ \widehat{\Ab^{-1}} = \big[\Acalh_1 \ \cdots \ \Acalh_k \big]$, of our approximation from Algorithm \ref{inv_alg}. We propose a linear encoding of the computed blocks $\{\Acalh_i\}_{i=1}^k$ based on generators satisfying \eqref{min_G_problem}. Along with the proposed decoding step, we have a MDS-based CCM for matrix inversion.

We consider the same parameters as in \ref{enc_A_subs}, in order to reuse $\Bcal$ in the proposed CMIM. Each $\Acalh_i$ is comprised of $T$ distinct but consecutive approximations of \eqref{inv_LS}, \textit{i.e.}
$$ \Acalh_i = \big[\bbh_{(i-1)T+1}\ \cdots \ \bbh_{iT}\big]\in\R^{N\times T} \quad \forall i\in\N_k, $$
which could also be approximated by iteratively solving
$$ \Acalh_i \approx \arg\hspace{-3mm}\min_{\ \ \Bb\in\R^{N\times T}} \hspace{-1mm} \left\{\left\|\Ab\Bb-\big[\eb_{(i-1)T+1}\ \cdots \ \eb_{iT}\big]\right\|_F^2\right\}. $$

Without loss of generality, we assume that the workers use the same algorithms and parameters for estimating the columns $\{\bbh_i\}_{i=1}^N$. Therefore, workers allocated the same tasks are expected to get equal approximations in the same amount of time.

For our CCM, we leverage $\brs$ generator matrices for both the encoding and decoding steps. We adapt the GC framework, so we need an analogous condition to $\ab_{\I}^\top\Gb_\I=\vec{\bold{1}}$ for the CMIM; in order to invoke Algorithm \ref{inv_alg}. The condition we require is $\Dbt_\I\Gbt_\I=\Ib_N$, for an encoding-decoding pair $(\Gbt,\Dbt_\I)$.

From our discussion on $\brs$ codes in \ref{BRS_CC_subs}, we set $\Gbt=\Ib_T\otimes\Gb$ and $\Dbt_\I=\Ib_T\otimes\Gb_\I^{-1}$ for any given set of $k$ responsive workers indexed by $\I$. The index set of blocks requested from the $\iota^{th}$ worker to compute is $\J_\iota\coloneqq\supp(\Gb_{(\iota)})$, and has cardinality $w$. The workers' encoding steps correspond to
\begin{equation}
\label{enc_identity}
  \Gbt\cdot(\widehat{\Ab^{-1}})^\top = (\Ib_T\otimes\Gb)\cdot\begin{bmatrix} \Acalh_1^\top \\ \vdots \\ \Acalh_k^\top \end{bmatrix} = \begin{pmatrix} \sum\limits_{j\in\J_1} p_j(\beta_1)\cdot\Acalh_j^\top \\ \vdots \\ \sum\limits_{j\in\J_n} p_j(\beta_n)\cdot\Acalh_j^\top \end{pmatrix}
\end{equation}
which are carried out locally, once they have computed their assigned tasks. We denote the encoding of the $\iota^{th}$ worker by $\Wb_\iota\in\C^{T\times N}$, \textit{i.e.} $\Wb_\iota = \sum_{j\in\J_\iota} p_j(\beta_\iota)\cdot\Acalh_j^\top$, which is sent to the coordinator. The received encoded computations by any distinct $k$ workers indexed by $\I$, constitute $\Gbt_\I\cdot(\widehat{\Ab^{-1}})^\top$.

Lemma \ref{inverse_lem} implies that as long as $k$ workers respond, the approximation $\widehat{\Ab^{-1}}$ is recoverable. Moreover, the decoding step reduces to a matrix multiplication of $k\times k$ matrices. Applying $\Hb_\I^{-1}$ to a square matrix can be done in $\ow(k^2\log k)$, through the IFFT algorithm. The prevailing computation in our decoding, is applying $\Pb^{-1}$. The decoding step is
{\small
\begin{align*}
  \Dbt_\I\cdot\left(\Gbt_\I\cdot(\widehat{\Ab^{-1}})^\top\right) &= \big(\Ib_T\otimes\Gb_\I^{-1}\big) \cdot\big(\Ib_T\otimes\Gb_\I\big)\cdot(\widehat{\Ab^{-1}})^\top\\
  &= (\Ib_T\cdot\Ib_T)\otimes\left(\Gb_\I^{-1}\cdot\Gb_\I\right)\cdot(\widehat{\Ab^{-1}})^\top \\
  &= \Ib_T\otimes\Ib_k\cdot(\widehat{\Ab^{-1}})^\top\\
  &= (\widehat{\Ab^{-1}})^\top
\end{align*}
}
and our scheme is valid.

The above CCM therefore has a linear encoding done locally by the workers \eqref{enc_identity}, is MDS since $s=d-1$, and its decoding step reduces to computing and applying $\Gb_\I^{-1}$ (Lemma \ref{inverse_lem}). The security of the encodings rely on the secrecy of $\Bcal$, which were sent from the coordinator to the workers. For an additional security layer, the interpolation points of $\Bcal$ could instead be defined as $\beta_{j}=\beta^{\pi(j)}$, for $\pi\in S_n$ a $\prp$. In this case, $\pi^{-1}$ would also need to be securely broadcasted.

\begin{figure}[h]
  \centering
    \includegraphics[scale=.2]{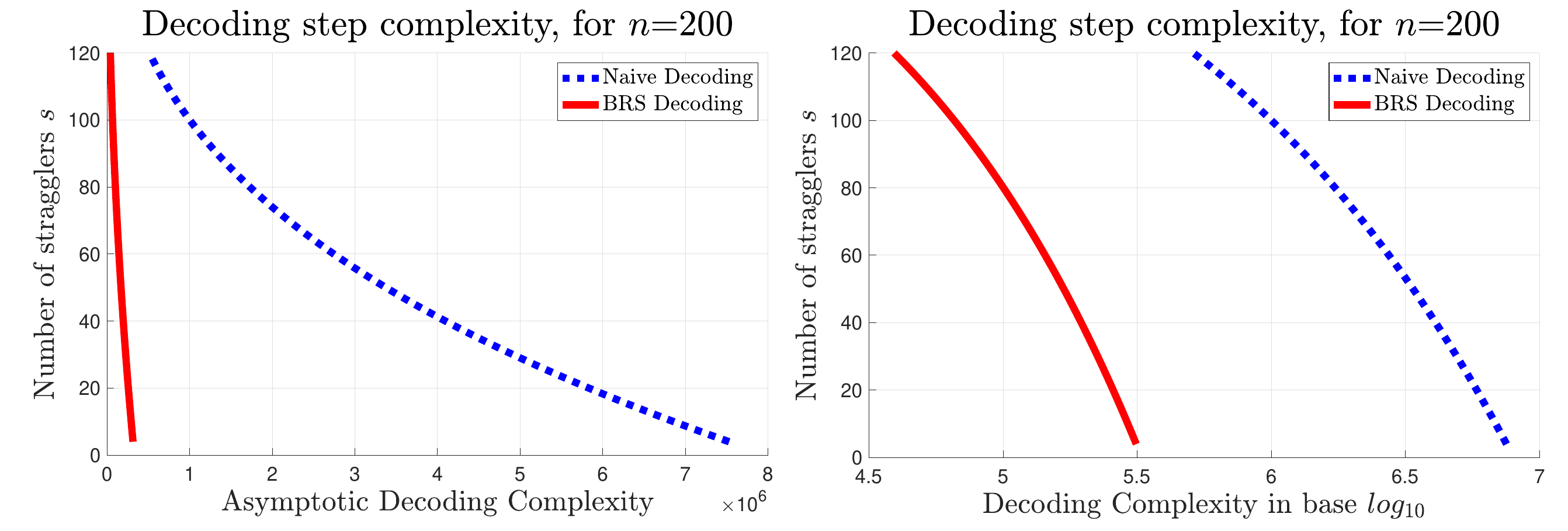}
    \caption{Comparison of decoding complexity, when naive matrix inversion is used (so $\ow(k^3)$) compared to the decoding step implied by Lemma \ref{inverse_lem}, for $n=200$ and varying $s$. We also provide a logarithmic scale comparison.}
  \label{decoding_compl_comparison}
\end{figure}

\begin{Rmk}
\label{rmk_lemma1}
With the above framework, any sparsest-balanced generator MDS matrix \cite{DSDY13} would suffice, as long as it satisfies the MDS theorem \cite{LX04}. By Lemma \ref{inverse_lem}, if we set $k=\Omega(\sqrt{N})$ (as in \cite{YMAA17}), the decoding step could then be done in $\ow\big(N^{\omega/2}\big)=o\big(N^{1.186}\big)$ time, which is close to linear in terms of $N$.
\end{Rmk}

\begin{Thm}
\label{MDS_CC_thm}
Let $\Gb\in\F^{n\times k}$ be a generator matrix of any MDS code over $\F$, for which $\|\Gb^{(j)}\|_0=n-k+1$ and $\|\Gb_{(i)}\|_0=w$ for all $(i,j)\in\N_n\times\N_k$. By utilizing Algorithm \ref{inv_alg}, we can devise a linear MDS coded matrix inversion scheme; through the encoding-decoding pair $(\Gbt,\Dbt_\I)$.
\end{Thm}

\begin{proof}
The encoding coefficients applied locally by each of the $n$ workers correspond to a row of $\Gb$. The encodings of all the workers then correspond to $\Gbt\cdot(\widehat{\Ab^{-1}})^\top$, for $\Gbt=\Ib_T\otimes\Gb$, as in \eqref{enc_identity}. Consider any set of responsive workers $\I$ of size $k$, whose encodings constitute $\Gbt_\I\cdot(\widehat{\Ab^{-1}})^\top$. By the MDS theorem, $\Gb_\I$ is invertible. Hence, the decoding step reduces to inverting $\Gb_\I$; \textit{i.e.} $\Dbt_\I=\Ib_T\otimes\Gb_\I^{-1}$, and is performed online.
\end{proof}

Constructions based on cyclic MDS codes, which have been used to devise GC schemes \cite{RTTD17}, can also be considered. These encoding matrices are not sparsest-balanced, which makes them suitable when considering \textit{heterogeneous} workers.

\begin{Prop}
\label{cyclic_MDS_prop}
Any cyclic $[n,k]$ MDS code $\Cc$ over $\F\in\{\R,\C\}$ can be used to devise a coded matrix inversion encoding-decoding pair $(\Gbt,\Dbt_\I)$.
\end{Prop}

\begin{proof}
Consider a cyclic $[n,n-s]$ MDS code $\Cc$ over $\F\in\{\R,\C\}$. Recall that from our assumptions, we have $s=n-k$. By \cite[Lemma 8]{RTTD17}, there exists a codeword $\gb_1\in\Cc$ of support $d=s+1$, \textit{i.e.} $\|\gb_1\|_0=d$. Since $\Cc$ is cyclic, it follows that the cyclic shifts of $\gb_1$ also lie in $\Cc$. Denote the $n-1$ consecutive cyclic shifts of $\gb_1$ by $\{\gb_i\}_{i=2}^n\subsetneq\Cc\subsetneq\F^{1\times n}$, which are all distinct. Define the cyclic matrix
\begin{equation*}
  \Gbb \coloneqq {\begin{pmatrix} | & | & & | \\ \gb_1^\top & \gb_2^\top & \hdots & \gb_n^\top \\ | & | & & | \end{pmatrix}} \in \F^{n\times n}.
\end{equation*}

Since $\|\gb_i\|_0=d$ and $\gb_i$ is a cyclic shift of $\gb_{i-1}$ for all $i>1$, it follows that $\|\Gbb_{(i)}\|_0=\|\Gbb_{(j)}\|_0=d$ for all $i,j\in\N_n$, \textit{i.e.} $\Gbb$ is sparsest and balanced. If we erase \textit{any} $s=n-k$ columns of $\Gbb$, we get $\Gb\in\F^{n\times k}$. By erasing arbitrary columns of $\Gbb$, the resulting $\Gb$ is \textit{not} balanced, \textit{i.e.} we have $\|\Gb_{(i)}\|_0\neq\|\Gb_{(j)}\|_0$ for some pairs $i,j\in\N_n$. Similar to our construction based on $\brs$ generator matrices, we define the encoding matrix to be $\Gbt=\Ib_T\otimes\Gb$. The local encodings are then analogous to \eqref{enc_identity}.

Consider an arbitrary set of $k$ non-straggling workers $\I\subsetneq\N_n$, and the corresponding matrix $\Gb_\I\in\F^{k\times k}$. By \cite[Lemma 12, B4.]{RTTD17}, $\Gb_\I$ is invertible. The decoding matrix is then $\Dbt_\I=\Ib_T\otimes\Gb_\I^{-1}$, and the condition $\Dbt_\I\Gbt=\Ib_N$ is met.
\end{proof}

\subsection{Optimality of MDS $\brs$ Codes}
\label{opt_BRS_subs}

Under the assumption that $k=n-s$, by utilizing the $\brs_q[n,k]$ generator matrices, we achieved the minimum possible communication load from the workers to the coordinator. From our discussion in \ref{know_of_A_subs}, we cannot hope to receive an encoding of less than $N^2/k$ symbols; when we require that $k$ workers respond with the same amount of information symbols in order to recover $\widehat{\Ab^{-1}}\in\R^{N\times N}$, unless we make further assumptions on the structure of $\Ab$ and $\Ab^{-1}$. Each encoding $\Wb_\iota$ consists of $NT=N^2/k$ symbols, so we have achieved the lower bound on the minimum amount of information needed to be sent to the coordinator. Moreover, $\Wb_\iota\in\C^{T\times N}$ for any sparsest-balance generator MDS matrix. This also holds true for other generator matrices which can be used in Theorem \ref{MDS_CC_thm}, as the encodings are linear (\textit{e.g.} Proposition \ref{cyclic_MDS_prop}).

We also require the workers to estimate the least possible number of columns for the given recovery threshold $k$. For our choice of parameters, the bound of \cite[Theorem 1]{TLDK17} is met with equality. That is, for all $i\in\N_n$:
$$ \|\Gb_{(i)}\|_0 = w = \frac{k}{n}\cdot d = \frac{k}{n}\cdot(n-k+1)\ , $$
which means that for homogeneous workers, we cannot get a sparser generator matrix. This, along with the requirement that $\Gb_\I$ should be invertible for all possible $\I$, are what we considered in \eqref{min_G_problem}.

\subsection{Time and Space Complexity}
\label{complexity_subsec}

Next, we discuss the complexity of our method. Communication loads and storage are measured in symbols over $\R$. For simplicity, we assume that the $n$ workers are homogeneous and the local data blocks $\{\Ab_\iota\}_{\iota=1}^k$ are of size $N\times T$, for $T=N/k=\Gamma\gamma$. To further simplify our expressions and for fair comparisons to other polynomial codes, we set $k=n=\Omega(\sqrt{N})$ (as in \cite{YMAA17}).

Let $\nu=|\Enc(\beta)|$ denote the number of symbols required for the encoding of $\beta$ through a public-key cryptosystem used to securely broadcast a symbol. Further note that $|\Enc(\Hh)|=\gamma\nu$ and $|\Enc(\sigma_\iota)|\leqslant\gamma\nu/2$, when the same cryptosystem is used to broadcast $\Hh$ and $\sigma_\iota$ respectively. Hence, phase $(\mathrm{a})$ requires a communication load of $\ow(\nu\gamma)$ symbols per client.

There are $\gamma$ modified Lagrange polynomials in \eqref{lagr_pol_matr}; each of which require $\ow(n-1)$ operations to compute. Multiplying each polynomial with one of the $\gamma$ sub-blocks $\{\Ab_\iota^j\}_{j=1}^\gamma$, requires $\gamma\cdot\ow(N\Gamma)=\ow(NT)=\ow(N^{3/2})$ additional operations in total. Finally, summing over the encoded blocks requires $\ow(N\Gamma(\gamma-1))=\ow(N^{3/2})$ operations. Hence, the encoding through the polynomials $f_\iota(\x)$ has complexity $\ow\big(2N^{3/2}+\gamma(n-1)\big)=\ow\big(\sqrt{N}(N+\gamma)\big)$ for each data block, which is done locally by the clients. The $\prp$ used could be any block cipher, \textit{e.g.} the Feistel cipher; which has time complexity $\ow(n)$. Therefore, phase $(\mathrm{b})$ has complexity $\ow\big(\sqrt{N}(N+\gamma+1)\big)$ and the clients communicate $2NT+n=\Omega\big(\sqrt{N}(N+1)\big)$ symbols to each other, which accounts for a total communication load of $(k-1)\cdot\Omega\big(\sqrt{N}(N+1)\big)=\Omega(N^2)$ symbols per client.

At phase $(\mathrm{c})$, the $\iota^{th}$ client first recovers $\Ab$ by evaluating $\eta_{\sigma_\iota(j)}^{-1}f_\iota(\beta_j)$; for each $j\in\N_\gamma$, which requires a total of $\ow(\gamma\cdot N\Gamma)=\ow\big(N^{3/2}\big)$ operations. The complexity of the computation tasks depends on the underlying optimization algorithm used by the workers; and the desired level of accuracy. By Proposition \ref{rel_err_bound_prop_CG}, under the given assumptions, when using CG we converge after $\tilde{N}$ iterations per column estimate, and each iteration has complexity $\ow(N)$. Therefore, the complexity of the workers' tasks are $\ow(\tilde{N}NT)=\ow\big(\tilde{N}N^{3/2}\big)$. All in all, the computation tasks at phase $(\mathrm{c})$ have total complexity $\ow\big((\tilde{N}+1)N^{3/2}\big)=\ow\big(\tilde{N}N^{3/2}\big)$ per worker. Since $\Wb_\iota\in\C^{T\times N}$ for each $\iota$, the communication load is $2NT=2N^2/k=\Omega\big(N^{3/2}\big)$ symbols. Furthermore, the baseline to computing the $\Ab^{-1}$ is Gaussian elimination; which has complexity $\ow(N^3)$ when carried out on one server, while our approach through CG has complexity $k\cdot\ow\big(\tilde{N}N^{3/2}\big)=\ow\big(\tilde{N}N^2\big)$.

By Lemma \ref{inverse_lem} and Remark \ref{rmk_lemma1}; the decoding takes $\ow(N^{\omega/2})=o(N^{1.186})$ time, which amounts to the complexity of phase $(\mathrm{d})$. Since our recovery threshold is $k$, phase $(\mathrm{d})$ no more than $k\cdot \Omega(NT)=\Omega(N^2)$ symbols need to be received and stored by the coordinator, who finally recovers a matrix of size $N\times N$.

We summarize the communication loads (C.L.) and time complexity (T.C.) of each of the four phases in the Table below. The time complexity for phase $(\mathrm{a})$ depends on the encryption method that is used to securely communicate $\beta,\Hh$; which we do not study, and there is no communication taking place in phase $(\mathrm{d})$. Phases $(\mathrm{a})$ and $(\mathrm{d})$ correspond to the coordinator, $(\mathrm{b})$ to each client, and $(\mathrm{c})$ to each worker.

\begin{center}
\begin{tabular}{ |p{.5cm}||p{1.71cm}|p{1.94cm}|p{1.22cm}|p{1.3cm}| }
\hline
\multicolumn{5}{|c|}{\textbf{Communication Loads $\&$ Time Complexities}} \\
\hline
\footnotesize{Phase} & {\footnotesize$(\mathrm{a})$ \textit{share} $\beta,\Hh$} & {\footnotesize$(\mathrm{b})$ \textit{encrypt} $\Ab_\iota$} & {\footnotesize$(\mathrm{c})$ \textit{CC job}} & {\footnotesize$(\mathrm{d})$ \textit{decode}} \\
\hhline{|=|=|=|=|=|}
\small{C.L.} & {\footnotesize$\ow(\nu\gamma)$} & {\footnotesize$\Omega(N^2)$} & {\footnotesize$\Omega\big(N^{3/2}\big)$} & \xmark \\
\hline
\small{T.C.} & \xmark & {\footnotesize$\ow\big(\sqrt{N}(N+\gamma)\big)$} & {\footnotesize$\ow\big(\tilde{N}N^{3/2}\big)$} & {\footnotesize$o\big(N^{1.186}\big)$} \\
\hline
\end{tabular}
\end{center}
\vspace{2mm}

A bottleneck of our approach is the workers' storage requirement. As was discussed in \ref{know_of_A_subs}, the workers need to recover $\Ab$, so they need to store a total of $N^2$ symbols. The central server receives a total of $k$ completed tasks $\{\Wb_i\}_{i\in\I}$, which constitute to a total of $2N^2$ symbols. Further examining this drawback would be worthwhile future work.

\subsection{Comparison to Exact Matrix Inversion}
\label{comp_exact_inversion}

We conclude this section with a discussion on the conditions under which our CMIM will have advantages over standard matrix inversion approaches. First of all, the main bottleneck of our approach is the fact that each worker has a storage requirement of $N^2$ symbols. When $N$ is relatively small, matrix inversion can be performed by a single server; though the time complexity is still high, in which case our distributed approach is beneficial. In this scenario, the storage constraint is not an issue. For $N$ very large, our approach is still advantageous in terms of time complexity, as a single server would need to perform the entire computation on its own; while also storing the entire matrix. In this case, the storage requirement of our approach is disadvantageous, since we require total storage of $kN^2$ symbols across the network, while matrix inversion only requires $N^2$. This is the cost we pay for performing our method distributively.

The second point of comparison is approximation accuracy. The accuracy of standard finite precision matrix inversion  is controlled by the number of bits of precision used by the multiplier. On the other hand, by design, our proposed algorithm introduces an additional approximation error due to its reliance on successive approximation iterations. As we showed numerically though in Figure \ref{err_MNIST_CG}, after a few iterations of Algorithm \ref{inv_alg} with CG; we can achieve the same error rate as when exact matrix inversion is used, with a lower complexity. Furthermore, in building risk minimizing ML models, approximate solutions using iterative approximations are often faster and sufficient for achieving desired performance benchmarks. Additionally, the approximation accuracy of our proposed matrix inversion method is controllable by adjusting the number of iterations carried out locally by the workers.

Lastly, we discuss when Algorithm \ref{inv_alg} might have advantages over exact matrix inversion in terms of computational complexity and waiting time. For simplicity, we assume that exact computation of $\Ab^{-1}$ requires $\ow(N^{2.372})$ operations. When utilizing our algorithm with CG on a single server, we require $\ow(\tilde{N}N^2)$ operations to guarantee convergence. Thus, in this case; Algorithm \ref{inv_alg} is beneficial when $\tilde{N}<N^{0.372}$, where $\tilde{N}$ is the number of distinct eigenvalues of $\Ab^\top\Ab$. When employing a distributed implementation, in terms of the waiting time through phase $(\mathrm{c})$; our approach is beneficial when $\tilde{N}<N^{0.872}$.

\vspace{-1mm}
\section{Conclusion and Future Work}
\label{concl_sec}

In this paper, we addressed the problem of approximate computation of the inverse of a matrix distributively in a relaxed FL setting, under the possible presence of straggling workers. We provided approximation error bounds for our approach, as well as security and recovery guarantees. We also provided numerical experiments that validated our proposed approach.

There are several interesting future directions. One avenue to consider is incorporating fully homomorphic encryption in our phases $(\mathrm{b})$,$(\mathrm{c})$,$(\mathrm{d})$, to obtain a FL scheme; and prevent the requirement of clients need to recover each others' information. An important issue is the numerical stability of the $\brs$ approach, so exploring other suitable generator matrices could be beneficial; \textit{e.g.} circulant permutation and rotation matrices \cite{RT21}. It is also worth investigating if we can reduce the communication rounds when computing the pseudoinverse through our approach. This depends on the CMM which is being utilized, though using different ones for each of the two multiplications may also be beneficial.

In terms of coding-theory, it would be interesting to see if it is possible to reduce the complexity of our decoding step. Specifically, could well-known $\rs$ decoding algorithms such as the Berlekamp-Welch algorithm be exploited? Another direction, is leveraging approximate CCMs. The work of \cite{JNMA21} considers the GC problem for \textit{approximate} and \textit{exact} recovery through Lagrange interpolation, for heterogeneous workers in the presence of stragglers and adversaries. A potential scheme for matrix inversion could also be developed through the methods of \cite{JNMA21}. In terms of our approximation algorithms, an avenue worth exploring is that of incorporating approximate and/or sparse Gaussian elimination \cite{KS16,Kyng17} into our distributed CCM.\\

\noindent\textbf{Tribute to Alex Vardy:} As this is a special issue dedicated to the memory Alexander Vardy, we mention how this paper relates to his work. Even though Alex had not worked on CC, his contributions to $\rs$ codes are immense. A focus of ours is to reduce the decoding complexity of the proposed $\brs$-based CCM, while in \cite{GV05} it was shown by Guruswami and Vardy that maximum-likelihood decoding of $\rs$ codes is \textsf{NP-hard}. Another highly innovative work of Vardy's is \cite{PV05}, in which the `Parvaresh-Vardy codes' were introduced; and the associated list-decoding algorithm was shown to yield an improvement over the Guruswami–Sudan algorithm. This was subsequently improved by Guruswami and Rudra \cite{GR08}, whose techniques were exploited in \cite{SLMA21} to introduce list-decoding in CC.\\

\noindent\textbf{Acknowledgements:} This work was partially supported by grants ARO W911NF-15-1-0479, DE NA0003921, NSF ECCS-2037304 and DMS-2134248, NSF CAREER Award CCF-2236829, U.S. ARO Early Career W911NF-21-1-0242, in part by the Stanford Precourt Institute; and the ACCESS -- AI Chip Center for Emerging Smart Systems through InnoHK, Hong Kong, SAR.

\appendices

\section{Additional Material and Background}

In this appendix, we include material and background which was used in our derivations. First, we recall what an \textit{$\epsilon$-optimal solution/point} is, which was used in the proof of Proposition \ref{rel_err_bound_prop_CG}. Next, we state the MDS theorem and the $\bch$ Bound. We then give a brief overview of the GC scheme from \cite{HASH17}, to show how it differs from our CMIM. We also explicitly give their construction of a balanced mask matrix $\Mb\in\{0,1\}^{n\times k}$, which we use for the construction of the $\brs$ generator matrices. Lastly, we illustrate a simple example of the encoding matrix.

\begin{Def}[\cite{BWZ08}]
A point $\bar{x}$ is said to be an \textbf{$\epsilon$-optimal solution/point} to a minimization problem with objective function $f(x)$, if for any $x$, it holds that $f(x)\geqslant f(\bar{x})-\epsilon$, where $\epsilon\geqslant0$. When $\epsilon=0$, an $\epsilon$-optimal solution is an exact minimizer.
\end{Def}

\begin{Thm}[MDS Theorem --- \cite{LX04}]
\label{MDS_thm}
\textit{Let $\Cc$ be a linear $[n,k,d]$ code over $\F_q$, with $\Gb,\Kb$ the generator and parity-check matrices. Then, the following are equivalent}:
\begin{enumerate}
  \item \textit{$\Cc$ is a MDS code, \textit{i.e.} $d=n-k+1$}
  \item \textit{every set of $n-k$ columns of $\Kb$ is linearly independent}
  \item \textit{every set of $k$ columns of $\Gb$ is linearly independent}
  \item \textit{$\Cc^{\perp}$ is a MDS code}.\\
\end{enumerate}
\end{Thm}

\begin{Thm}[$\bch$ Bound --- \cite{HLH16},\cite{Mc01}]
\label{BCH_bd}
Let $p(\x)\in\F_q[\x]\backslash\{0\}$ with $t$ cyclically consecutive roots, \textit{i.e.} $p(\alpha^{j+\iota})=0$ for all $\iota\in\N_t$. Then, at least $t+1$ coefficients of $p(\x)$ are nonzero.
\end{Thm}

\begin{algorithm}[h]
\label{RBMM}
\SetAlgoLined
\KwIn{$n,k,d\in\Z_+$ s.t. $n>d,k$ and $w=\frac{kd}{n}$}
\KwOut{row-balanced mask matrix $\Mb\in\{0,1\}^{n\times k}$}
  $\Mb \gets \bold{0}_{n\times k}$\\
  \For{$j = 0$ to $k-1$}
    {
    \For{$i = 0$ to $d-1$}
      {
        $\iota \gets (i+jd+1)\bmod n$\\
        $\Mb_{r,\iota} \gets 1$
      }
    }
 \Return $\Mb$
 \caption{$\mathrm{MaskMatrix}(n,k,d)$ \cite{HASH17}}
\end{algorithm}

\subsection{Generator Matrix Example}

For an example, consider the case where $n=9$, $k=6$ and $d=6$, thus $w=\frac{kd}{n}=4$. Then, Algorithm \ref{RBMM} produces 
$$ \Mb = \begin{pmatrix} 1 & 1 & {\gray 0} & 1 & 1 & {\gray 0} \\ 1 & 1 & {\gray 0} & 1 & 1 & {\gray 0} \\ 1 & 1 & {\gray 0} & 1 & 1 & {\gray 0} \\ 1 & {\gray 0} & 1 & 1 & {\gray 0} & 1 \\ 1 & {\gray 0} & 1 & 1 & {\gray 0} & 1 \\ 1 & {\gray 0} & 1 & 1 & {\gray 0} & 1 \\ {\gray 0} & 1 & 1 & {\gray 0} & 1 & 1 \\ {\gray 0} & 1 & 1 & {\gray 0} & 1 & 1 \\ {\gray 0} & 1 & 1 & {\gray 0} & 1 & 1 \end{pmatrix} \ \in\{0,1\}^{9\times 6} \ . $$
For our CCM, this means that the $i^{th}$ worker computes the blocks indexed by $\supp(\Mb_{(i)})$, \textit{e.g}. $\supp(\Mb_{(1)})=\{1,2,4,5\}$. We denote the indices of the respective task allocations by $\J_i=\supp(\Mb_{(i)})$. The entries of the generator matrix $\Gb$ are the evaluations of the constructed polynomials \eqref{lagr_polys} at each of the  evaluation points $\Bcal=\{\beta_i\}_{i=1}^n$, \textit{i.e.} $\Gb_{ij}=p_j(\beta_i)$. This results in:
$$ \Gb = \begin{pmatrix} p_1({\beta_1}) & p_2({\beta_1}) & {\gray 0} & p_4({\beta_1}) & p_5({\beta_1}) & {\gray 0} \\ p_1({\beta_2}) & p_2({\beta_2}) & {\gray 0} & p_4({\beta_2}) & p_5({\beta_2}) & {\gray 0} \\ p_1({\beta_3}) & p_2({\beta_3}) & {\gray 0} & p_4({\beta_3}) & p_5({\beta_3}) & {\gray 0} \\  p_1({\beta_4}) & {\gray 0} & p_3({\beta_4}) & p_4({\beta_4}) & {\gray 0} & p_6({\beta_4}) \\  p_1({\beta_5}) & {\gray 0} & p_3({\beta_5}) & p_4({\beta_5}) & {\gray 0} & p_6({\beta_5}) \\  p_1({\beta_6}) & {\gray 0} & p_3({\beta_6}) & p_4({\beta_6}) & {\gray 0} & p_6({\beta_6}) \\ {\gray 0} & p_2({\beta_7}) & p_3({\beta_7}) & {\gray 0} & p_5({\beta_7}) & p_6({\beta_7}) \\ {\gray 0} & p_2({\beta_8}) & p_3({\beta_8}) & {\gray 0} & p_5({\beta_8}) & p_6({\beta_8}) \\ {\gray 0} & p_2({\beta_9}) & p_3({\beta_9}) & {\gray 0} & p_5({\beta_9}) & p_6({\beta_9}) \\ \end{pmatrix} . $$


\section{Distributed Pseudoinverse}

For full-rank rectangular matrices $\Ab\in\R^{N\times M}$ where $N>M$, one resorts to the left Moore–Penrose pseudoinverse $\Ab^{\dagger}\in\R^{M\times N}$, for which $\Ab^{\dagger}\Ab=\Ib_M$. In Algorithm \ref{pinv_alg}, we present how to approximate the left pseudoinverse of $\Ab$, by using the fact that $\Ab^{\dagger}=(\Ab^\top\Ab)^{-1}\Ab^\top$; since $\Ab^\top\Ab\in\GL_M(\R)$. The right pseudoinverse $\Ab^{\dagger}=\Ab^\top(\Ab\Ab^\top)^{-1}$ of $\Ab\in\R^{M\times N}$ where $M<N$, can be obtained by a modification of Algorithm \ref{pinv_alg}.

Just like the inverse, the pseudoinverse of a matrix also appears in a variety of applications. Computing the pseudoinverse of $\Ab\in\R^{N\times M}$ for $N>M$ is even more cumbersome, as it requires inverting the Gram matrix $\Ab^\top\Ab$. For this appendix, we consider a full-rank matrix $\Ab$.

One could naively attempt to modify Algorithm \ref{inv_alg} in order to retrieve $\Ab^{\dagger}$ such that $\Ab^{\dagger}\Ab=\Ib_M$, by approximating the rows of $\Ab^{\dagger}$. This would \textit{not} work, as the underlying optimization problems would not be strictly convex. Instead, we use Algorithm \ref{pinv_alg} to estimate the rows of $\Bb^{-1}\coloneqq(\Ab^\top\Ab)^{-1}$, and then multiply the estimate $\widehat{\Bb^{-1}}$ by $\Ab^\top$. This gives us the approximation $\widehat{\Ab^{\dagger}}=\widehat{\Bb^{-1}}\cdot\Ab^\top$.

The drawback of Algorithm \ref{pinv_alg} is that it requires two additional matrix multiplications, $\Ab^\top\Ab$ and $\widehat{\Bb^{-1}}\Ab^\top$. We overcome this barrier by using a CMM scheme twice, to recover $\widehat{\Ab^{\dagger}}$ in a two or three-round communication CC approach. These are discussed in below.

Bounds on $\text{err}_F(\widehat{\Ab^{-1}})$ and $\text{err}_{\text{r}F}(\widehat{\Ab^{-1}})$ can be established for both algorithms, specific to the black-box least squares solver being utilized.

\begin{algorithm}[h]
\label{pinv_alg}
\SetAlgoLined
\KwIn{full-rank $\Ab\in\R^{N\times M}$ where $N>M$}
  $\Bb\gets \Ab^\top\Ab$\\
  \For{i=1 to M}
  {
    $\hat{\cb}_i = \arg\min_{\cb\in\R^{1\times M}} \Big\{g_i(\cb)\coloneqq\|\cb\Bb-\eb_i^\top\|_2^2\Big\}$\\ 
    $\hat{\bb}_i\gets \hat{\cb}_i\cdot\Ab^\top$
  }
 \Return $\widehat{\Ab^{\dagger}} \gets \left[ \hat{\bb}_1^\top \ \cdots \ \hat{\bb}_M^\top \right]^\top$ \Comment{$\widehat{\Ab^{\dagger}}_{(i)}=\bbh_i$}
 \caption{Estimating $\Ab^{\dagger}$}
\end{algorithm}

\begin{Cor}
\label{rel_err_bound_cor}
  For full-rank $\Ab\in\R^{N\times M}$ with $N>M$, we have $\err_{F}(\widehat{\Ab^{\dagger}})\leqslant\frac{\sqrt{M}\epsilon\cdot\kappa_2}{\sqrt{2}\sigma_{\mathrm{min}}(\Ab)^3}$ and $\err_{\mathrm{r}F}(\widehat{\Ab^{\dagger}})\leqslant\frac{\sqrt{M}\epsilon\cdot\kappa_2}{\sqrt{2}\sigma_{\mathrm{min}}(\Ab)^2}$ when using SD to solve the subroutine optimization problems of Algorithm \ref{pinv_alg}, with termination criteria $\|\nabla g_i(\cb^{[t]})\|_2\leqslant\epsilon$.
\end{Cor}

\begin{proof}
From \eqref{deriv_rF_bd}, it follows that
$$ \|\Bb^{-1}\eb_i-\hat{\cb}_i^\top\|_2 \leqslant \frac{\epsilon/\sqrt{2}}{\sigma_{\mathrm{min}}(\Bb)^2} = \frac{\epsilon/\sqrt{2}}{\sigma_{\mathrm{min}}(\Ab)^4} \eqqcolon \delta \ . $$
The above bound implies that for each summand of the Frobenius error;
$\|\hat{\bb}_i-\Ab^{\dagger}_{(i)}\|_2=\|\hat{\cb}_i\Ab^\top-\eb_i^\top\cdot\Bb^{-1}\Ab^\top\|_2$, we have $\|\hat{\bb}_i-\Ab^{\dagger}_{(i)}\|_2\leqslant\delta\|\Ab^\top\|_2$. Summing the right hand side $M$ times, we get that
\begin{align*}
  \err_{F}(\widehat{\Ab^{\dagger}})^2 &\leqslant M\cdot(\delta\|\Ab^\top\|_2)^2\\
  &= \frac{M\epsilon^2\cdot\sigma_\text{max}(\Ab)^2}{2\sigma_\text{min}(\Ab)^8}\\
  &= \frac{M\epsilon^2\cdot\kappa_2^2}{2\sigma_\text{min}(\Ab)^6}\ .
\end{align*}
By taking the square root, we have shown the first claim.

Since $1/\sigma_{\text{min}}(\Ab)=\|\Ab^{\dagger}\|_2\leqslant\|\Ab^{\dagger}\|_F$, it then follows that
$$ \err_{\mathrm{r}F}(\widehat{\Ab^{\dagger}}) = \frac{\err_{F}(\widehat{\Ab^{\dagger}})}{\|\Ab^{\dagger}\|_F} \leqslant \frac{\err_{F}(\widehat{\Ab^{\dagger}})}{\|\Ab^{\dagger}\|_2} = \frac{\sqrt{M}\epsilon\cdot\kappa_2}{\sqrt{2}\sigma_{\mathrm{min}}(\Ab)^2}\ , $$
which completes the proof.
\end{proof}

\subsection{Pseudoinverse from Polynomial CMM}
\label{distr_Pseudinv_subsec}

One approach to leverage Algorithm \ref{pinv_alg} in a two-round communication scheme is to first compute $\Bb=\Ab^\top\Ab$ through a CMM scheme, then share $\Bb$ with all the workers who estimate the rows of $\widehat{\Bb^{-1}}$, and finally use another CMM to locally encode the estimated columns with blocks of $\Ab^\top$; to recover $\widehat{\Ab^{\dagger}}=\widehat{\Bb^{-1}}\cdot\Ab^\top$. Even though there are only two rounds of communication, the fact that we have a local encoding by the workers results in a higher communication load overall. An alternative approach which circumvents this issue, uses three-rounds of communication.

For this approach, we use the polynomial CMM scheme from \cite{YMAA17} twice, along with our coded matrix inversion scheme. This CMM has a reduced communication load, and minimal computation is required by the workers. To have a consistent recovery threshold across our communication rounds, we partition $\Ab$ as in \eqref{parts_A_iota} into $\kbar=\sqrt{n-s}=\sqrt{k}$ blocks. Each block is of size $N\times\Tbar$, for $\Tbar=\frac{M}{k}$. The encodings from \cite{YMAA17} of the partitions $\{\Ab_\iota\}_{\iota=1}^{\kbar}$ for carefully selected parameters $a,b\in\Z_+$ and distinct elements $\gamma_i\in\F_q$, are
$$ \Abt^a_i=\sum_{j=1}^k\Ab_j\gamma_i^{(j-1)a} \quad \text{ and } \quad \Abt^b_i=\sum_{j=1}^k\Ab_j\gamma_i^{(j-1)b} $$
for each worker indexed by $i$. Thus, each encoding is comprised of $N\Tbar$ symbols. The workers compute the product of their respective encodings $(\Abt^a_i)^\top\cdot\Abt^b_i$. The decoding step corresponds to an interpolation step, which is achievable when $\kbar^2=k$ many workers respond\footnote{We select $\kbar=\sqrt{k}$ in the partitioning of $\Ab$ in \eqref{parts_A_iota} when deploying this CMM, to attain the same recovery threshold as our inversion scheme.}, which is the optimal recovery threshold for CMM. Any fast polynomial interpolation or $\rs$ decoding algorithm can be used for this step, to recover $\Bb$.

Next, the master shares $\Bb$ with all the workers (from \ref{know_of_A_subs}, this is necessary), who are requested to estimate the \textit{column-blocks} of $\widehat{\Bb^{-1}}$
\begin{equation}
\label{parts_Binv}
  \widehat{\Bb^{-1}}=\Big[\Bcalb_1 \ \cdots \ \Bcalb_k\Big] \ \ \text{ where } \Bcalb_j\in\R^{M\times\Tbar}\ \forall j\in\N_k
\end{equation}
according to Algorithm \ref{inv_alg}. We can then recover $\widehat{\Bb^{-1}}$ by our $\brs$ based scheme, once $k$ workers send their encoding.

For the final round, we encode $\widehat{\Bb^{-1}}$ as 
$$ \Bbt^a_i=\sum_{j=1}^k\Bcalb_j\gamma_i^{(j-1)a} $$
which are sent to the respective workers. The workers already have in their possession the encodings $\Abt^b_i$. We then carry out the polynomial CMM where each worker is requested to send back $(\Bbt^a_i)^\top\cdot\Abt^b_i$. The master server can then recover $\widehat{\Ab^{\dagger}}$.

\begin{Thm}
\label{MDS_CC_psinv_thm}
Consider $\Gb\in\F^{n\times k}$ as in Theorem \ref{MDS_CC_thm}. By using any CMM, we can devise a matrix pseudoinverse CCM by utilizing Algorithm \ref{pinv_alg}, in two-rounds of communication. By using polynomial CMM \cite{YMAA17}, we achieve this with a reduced communication load and minimal computation, in three-rounds of communication.
\end{Thm}


\bibliographystyle{IEEEtran}
\bibliography{refs}

\begin{thebibliography}{10}
\providecommand{\url}[1]{#1}
\csname url@samestyle\endcsname
\providecommand{\newblock}{\relax}
\providecommand{\bibinfo}[2]{#2}
\providecommand{\BIBentrySTDinterwordspacing}{\spaceskip=0pt\relax}
\providecommand{\BIBentryALTinterwordstretchfactor}{4}
\providecommand{\BIBentryALTinterwordspacing}{\spaceskip=\fontdimen2\font plus
\BIBentryALTinterwordstretchfactor\fontdimen3\font minus
  \fontdimen4\font\relax}
\providecommand{\BIBforeignlanguage}[2]{{%
\expandafter\ifx\csname l@#1\endcsname\relax
\typeout{** WARNING: IEEEtran.bst: No hyphenation pattern has been}%
\typeout{** loaded for the language `#1'. Using the pattern for}%
\typeout{** the default language instead.}%
\else
\language=\csname l@#1\endcsname
\fi
#2}}
\providecommand{\BIBdecl}{\relax}
\BIBdecl

\bibitem{GS59}
B.~G. Greenberg and A.~E. Sarhan, ``Matrix inversion, its interest and
  application in analysis of data,'' \emph{Journal of the American Statistical
  Association}, vol.~54, no. 288, pp. 755--766, 1959.

\bibitem{Hig02}
N.~J. Higham, \emph{Accuracy and Stability of Numerical Algorithms},
  2nd~ed.\hskip 1em plus 0.5em minus 0.4em\relax USA: Society for Industrial
  and Applied Mathematics, 2002.

\bibitem{LLPPR17}
K.~Lee, M.~Lam, R.~Pedarsani, D.~Papailiopoulos, and K.~Ramchandran, ``Speeding
  up distributed machine learning using codes,'' \emph{IEEE Transactions on
  Information Theory}, vol.~64, no.~3, pp. 1514--1529, 2017.

\bibitem{YSRKSA18}
Q.~Yu, S.~Li, N.~Raviv, S.~M.~M. Kalan, M.~Soltanolkotabi, and S.~A.
  Avestimehr, ``{Lagrange Coded Computing: Optimal Design for Resiliency,
  Security, and Privacy},'' in \emph{The 22nd International Conference on
  Artificial Intelligence and Statistics}.\hskip 1em plus 0.5em minus
  0.4em\relax PMLR, 2019, pp. 1215--1225.

\bibitem{KSD17}
C.~Karakus, Y.~Sun, and S.~Diggavi, ``{Encoded Distributed Optimization},'' in
  \emph{2017 IEEE International Symposium on Information Theory (ISIT)}.\hskip
  1em plus 0.5em minus 0.4em\relax IEEE, 2017, pp. 2890--2894.

\bibitem{CMPH22}
N.~Charalambides, H.~Mahdavifar, M.~Pilanci, and A.~O. Hero, ``{Orthonormal
  Sketches for Secure Coded Regression},'' in \emph{2022 IEEE International
  Symposium on Information Theory (ISIT)}, 2022, pp. 826--831.

\bibitem{LA20}
S.~Li and S.~Avestimehr, ``{Coded Computing},'' \emph{Foundations and
  Trends{\textregistered} in Communications and Information Theory}, vol.~17,
  no.~1, 2020.

\bibitem{YMAA17}
Q.~Yu, M.~Maddah-Ali, and S.~Avestimehr, ``{Polynomial Codes: an Optimal Design
  for High-Dimensional Coded Matrix Multiplication},'' in \emph{Advances in
  Neural Information Processing Systems}, 2017, pp. 4403--4413.

\bibitem{CPH22a}
N.~Charalambides, M.~Pilanci, and A.~O. Hero, ``{Secure Linear MDS Coded Matrix
  Inversion},'' in \emph{2022 58th Annual Allerton Conference on Communication,
  Control, and Computing (Allerton)}, 2022, pp. 1--8.

\bibitem{HASH17}
W.~Halbawi, N.~Azizan, F.~Salehi, and B.~Hassibi, ``{Improving Distributed
  Gradient Descent Using Reed-Solomon Codes},'' in \emph{2018 IEEE
  International Symposium on Information Theory (ISIT)}.\hskip 1em plus 0.5em
  minus 0.4em\relax IEEE, 2018, pp. 2027--2031.

\bibitem{DPYTH19}
S.~Dhakal, S.~Prakash, Y.~Yona, S.~Talwar, and N.~Himayat, ``{Coded Federated
  Learning},'' in \emph{2019 IEEE Globecom Workshops (GC Wkshps)}.\hskip 1em
  plus 0.5em minus 0.4em\relax IEEE, 2019, pp. 1--6.

\bibitem{PDAYTAH20}
S.~Prakash, S.~Dhakal, M.~R. Akdeniz, Y.~Yona, S.~Talwar, S.~Avestimehr, and
  N.~Himayat, ``{Coded Computing for Low-Latency Federated Learning over
  Wireless Edge Networks},'' \emph{IEEE Journal on Selected Areas in
  Communications}, vol.~39, no.~1, pp. 233--250, 2020.

\bibitem{SKRA21}
R.~Schlegel, S.~Kumar, E.~Rosnes, and A.~G.~i. Amat, ``{CodedPaddedFL and
  CodedSecAgg: Straggler Mitigation and Secure Aggregation in Federated
  Learning},'' \emph{arXiv e-prints}, pp. arXiv--2112, 2021.

\bibitem{SRRA22}
S.~Kumar, R.~Schlegel, E.~Rosnes, and A.~G.~i. Amat, ``{Coding for Straggler
  Mitigation in Federated Learning},'' \emph{arXiv preprint arXiv:2109.15226},
  2021.

\bibitem{XARWZ22}
M.~Xhemrishi, A.~G.~i. Amat, E.~Rosnes, and A.~Wachter-Zeh, ``{Computational
  Code-Based Privacy in Coded Federated Learning},'' \emph{arXiv preprint
  arXiv:2202.13798}, 2022.

\bibitem{HZSK19}
S.~Ha, J.~Zhang, O.~Simeone, and J.~Kang, ``{Coded Federated Computing in
  Wireless Networks with Straggling Devices and Imperfect CSI},'' in \emph{2019
  IEEE International Symposium on Information Theory (ISIT)}, 2019, pp.
  2649--2653.

\bibitem{HLH16}
W.~Halbawi, Z.~Liu, and B.~Hassibi, ``Balanced {R}eed-{S}olomon {C}odes,'' in
  \emph{2016 IEEE International Symposium on Information Theory (ISIT)}.\hskip
  1em plus 0.5em minus 0.4em\relax IEEE, 2016, pp. 935--939.

\bibitem{HLH16b}
------, ``Balanced {R}eed-{S}olomon {C}odes for all parameters,'' in \emph{2016
  IEEE Information Theory Workshop (ITW)}.\hskip 1em plus 0.5em minus
  0.4em\relax IEEE, 2016, pp. 409--413.

\bibitem{CMH20}
N.~Charalambides, H.~Mahdavifar, and A.~O. Hero, ``{Numerically Stable Binary
  Gradient Coding},'' \emph{arXiv preprint arXiv:2001.11449}, 2020.

\bibitem{KMRR16}
J.~Kone{\v{c}}n{\`y}, H.~B. McMahan, D.~Ramage, and P.~Richt{\'a}rik,
  ``{Federated Optimization: Distributed Machine Learning for On-Device
  Intelligence},'' \emph{arXiv preprint arXiv:1610.02527}, 2016.

\bibitem{YGK17}
Y.~Yang, P.~Grover, and S.~Kar, ``{Coded Distributed Computing for Inverse
  Problems},'' in \emph{Advances in Neural Information Processing Systems},
  vol.~30.\hskip 1em plus 0.5em minus 0.4em\relax Curran Associates, Inc.,
  2017, pp. 709--719.

\bibitem{TLDK17}
R.~Tandon, Q.~Lei, A.~G. Dimakis, and N.~Karampatziakis, ``Gradient coding:
  Avoiding stragglers in distributed learning,'' in \emph{International
  Conference on Machine Learning}, 2017, pp. 3368--3376.

\bibitem{SHPN20}
A.~M. Subramaniam, A.~Heidarzadeh, A.~K. Pradhan, and K.~R. Narayanan,
  ``{Product Lagrange Coded Computing},'' in \emph{2020 IEEE International
  Symposium on Information Theory (ISIT)}, 2020, pp. 197--202.

\bibitem{FC19b}
M.~Fahim and V.~R. Cadambe, ``{Lagrange Coded Computing with Sparsity
  Constraints},'' in \emph{2019 57th Annual Allerton Conference on
  Communication, Control, and Computing (Allerton)}, 2019, pp. 284--289.

\bibitem{SMA21}
M.~Soleymani, H.~Mahdavifar, and A.~S. Avestimehr, ``{Analog Lagrange Coded
  Computing},'' \emph{IEEE Journal on Selected Areas in Information Theory},
  vol.~2, no.~1, pp. 283--295, 2021.

\bibitem{SLMA21}
M.~Soleymani, R.~E. Ali, H.~Mahdavifar, and A.~S. Avestimehr, ``{List-Decodable
  Coded Computing: Breaking the Adversarial Toleration Barrier},'' \emph{IEEE
  Journal on Selected Areas in Information Theory}, vol.~2, no.~3, pp.
  867--878, 2021.

\bibitem{ZL22}
J.~Zhu and S.~Li, ``{Generalized Lagrange Coded Computing: A Flexible
  Computation-Communication Tradeoff},'' in \emph{2022 IEEE International
  Symposium on Information Theory (ISIT)}, 2022, pp. 832--837.

\bibitem{KD22}
S.~Kiani and S.~C. Draper, ``{Successive Approximation Coding for Distributed
  Matrix Multiplication},'' \emph{IEEE Journal on Selected Areas in Information
  Theory}, vol.~3, no.~2, pp. 286--305, 2022.

\bibitem{DFHJCG19}
S.~Dutta, M.~Fahim, F.~Haddadpour, H.~Jeong, V.~Cadambe, and P.~Grover, ``{On
  the Optimal Recovery Threshold of Coded Matrix Multiplication},'' \emph{IEEE
  Transactions on Information Theory}, vol.~66, no.~1, pp. 278--301, 2019.

\bibitem{YMAA20}
Q.~Yu, M.~A. Maddah-Ali, and A.~S. Avestimehr, ``{Straggler Mitigation in
  Distributed Matrix Multiplication: Fundamental Limits and Optimal Coding},''
  \emph{IEEE Transactions on Information Theory}, vol.~66, no.~3, pp.
  1920--1933, 2020.

\bibitem{DSDY13}
S.~H. Dau, W.~Song, Z.~Dong, and C.~Yuen, ``{Balanced Sparsest Generator
  Matrices for MDS Codes},'' in \emph{2013 IEEE International Symposium on
  Information Theory}, 2013, pp. 1889--1893.

\bibitem{Kra96}
M.~Krause, ``{A Simple Proof of the Gale-Ryser Theorem},'' \emph{The American
  Mathematical Monthly}, vol. 103, no.~4, pp. 335--337, 1996.

\bibitem{Mc01}
R.~J. McEliece, \emph{Theory of Information and Coding}, 2nd~ed.\hskip 1em plus
  0.5em minus 0.4em\relax USA: Cambridge University Press, 2001.

\bibitem{WXXZ23}
V.~V. Williams, Y.~Xu, Z.~Xu, and R.~Zhou, ``{New Bounds for Matrix
  Multiplication: from Alpha to Omega},'' \emph{arXiv preprint
  arXiv:2307.07970}, 2023.

\bibitem{BP70}
{\AA}.~Bj{\"o}rck and V.~Pereyra, ``{Solution of Vandermonde Systems of
  Equations},'' \emph{Mathematics of Computation}, vol.~24, pp. 893--903, 1970.

\bibitem{BV04}
S.~P. Boyd and L.~Vandenberghe, \emph{{Convex Optimization}}.\hskip 1em plus
  0.5em minus 0.4em\relax Cambridge university press, 2004.

\bibitem{She94}
J.~R. Shewchuk, ``{An Introduction to the Conjugate Gradient Method Without the
  Agonizing Pain},'' \emph{Carnegie Mellon University, Tech. Rep.}, 1994.

\bibitem{TB97}
L.~N. Trefethen and D.~Bau~III, \emph{Numerical linear algebra}.\hskip 1em plus
  0.5em minus 0.4em\relax Siam, 1997, vol.~50.

\bibitem{Bub15}
\BIBentryALTinterwordspacing
S.~Bubeck, ``{Convex Optimization: Algorithms and Complexity},''
  \emph{Foundations and Trends® in Machine Learning}, vol.~8, no. 3-4, pp.
  231--357, 2015. [Online]. Available:
  \url{http://dx.doi.org/10.1561/2200000050}
\BIBentrySTDinterwordspacing

\bibitem{Gow16}
R.~M. Gower, ``{Sketch and Project: Randomized Iterative Methods for Linear
  Systems and Inverting Matrices},'' \emph{arXiv preprint arXiv:1612.06013},
  2016.

\bibitem{CT06}
T.~M. Cover and J.~A. Thomas, \emph{Elements of Information Theory (Wiley
  Series in Telecommunications and Signal Processing)}.\hskip 1em plus 0.5em
  minus 0.4em\relax USA: Wiley-Interscience, 2006.

\bibitem{Sha79}
A.~Shamir, ``{How to Share a Secret},'' \emph{Communications of the ACM},
  vol.~22, no.~11, pp. 612--613, 1979.

\bibitem{Bla79}
G.~R. Blakley, ``Safeguarding cryptographic keys,'' \emph{1979 International
  Workshop on Managing Requirements Knowledge (MARK)}, pp. 313--318, 1899.

\bibitem{LX04}
S.~Ling and C.~Xing, \emph{{Coding Theory: A First Course}}.\hskip 1em plus
  0.5em minus 0.4em\relax Cambridge University Press, 2004.

\bibitem{RTTD17}
N.~Raviv, I.~Tamo, R.~Tandon, and A.~G. Dimakis, ``{Gradient Coding from Cyclic
  MDS Codes and Expander Graphs},'' \emph{IEEE Transactions on Information
  Theory}, vol.~66, no.~12, pp. 7475--7489, 2020.

\bibitem{RT21}
A.~Ramamoorthy and L.~Tang, ``Numerically stable coded matrix computations via
  circulant and rotation matrix embeddings,'' \emph{IEEE Transactions on
  Information Theory}, vol.~68, no.~4, pp. 2684--2703, 2021.

\bibitem{JNMA21}
T.~Jahani-Nezhad and M.~A. Maddah-Ali, ``{Optimal Communication-Computation
  Trade-Off in Heterogeneous Gradient Coding},'' \emph{IEEE Journal on Selected
  Areas in Information Theory}, vol.~2, no.~3, pp. 1002--1011, 2021.

\bibitem{KS16}
R.~Kyng and S.~Sachdeva, ``{Approximate Gaussian Elimination for Laplacians
  –- Fast, Sparse, and Simple},'' in \emph{2016 IEEE 57th Annual Symposium on
  Foundations of Computer Science (FOCS)}.\hskip 1em plus 0.5em minus
  0.4em\relax IEEE, 2016, pp. 573--582.

\bibitem{Kyng17}
R.~Kyng, ``{Approximate Gaussian Elimination},'' Ph.D. dissertation, PhD
  thesis. Yale University, 2017.

\bibitem{GV05}
V.~Guruswami and A.~Vardy, ``{Maximum-Likelihood Decoding of Reed-Solomon Codes
  is NP-hard},'' \emph{IEEE Transactions on Information Theory}, vol.~51,
  no.~7, pp. 2249--2256, 2005.

\bibitem{PV05}
F.~Parvaresh and A.~Vardy, ``{Correcting Errors Beyond the Guruswami-Sudan
  Radius in Polynomial Time},'' in \emph{46th Annual IEEE Symposium on
  Foundations of Computer Science (FOCS'05)}.\hskip 1em plus 0.5em minus
  0.4em\relax IEEE, 2005, pp. 285--294.

\bibitem{GR08}
V.~Guruswami and A.~Rudra, ``{Explicit Codes Achieving List Decoding Capacity:
  Error-Correction With Optimal Redundancy},'' \emph{IEEE Transactions on
  Information Theory}, vol.~54, no.~1, pp. 135--150, 2008.

\bibitem{BWZ08}
F.~Bai, Z.~Wu, and D.~Zhu, ``{Sequential Lagrange multiplier condition for
  $\epsilon$-optimal solution in convex programming},'' \emph{Optimization},
  vol.~57, no.~5, pp. 669--680, 2008.

\end{thebibliography}
\balance


\end{document}